\documentclass[11pt,a4paper]{article}

\usepackage{fullpage}
\usepackage[english]{babel}
\usepackage[T1]{fontenc}
\usepackage{times}
\usepackage[hidelinks]{hyperref}
\usepackage{ifthen}
\usepackage{amsmath,amsthm,amssymb}

\usepackage{float}
\usepackage{xspace}
\usepackage[utf8]{inputenc}
\usepackage{calc,slashed,graphicx}
\usepackage[usenames,dvipsnames]{xcolor}
\usepackage{tabularx}
\usepackage{etoolbox}

\usepackage{fancybox}
\usepackage{tikz}

\usetikzlibrary{math}
\usetikzlibrary{arrows.meta}
\usetikzlibrary{patterns}
\usetikzlibrary{calc}
\usetikzlibrary{decorations.pathreplacing}
\usetikzlibrary{shapes.callouts}

\input{fig/tikz_callouts_pointer_shift}

\usepackage[vlined,boxed,linesnumbered]{algorithm2e}
\SetAlgoVlined

\SetKwHangingKw{InitStep}{Init:}{}{}
\SetKwProg{ForAll}{for all}{ do}{end}
\SetKwProg{ForEach}{for each}{ do}{end}
\SetKwProg{InEachSynchroRound}{In each synchronous round}{ do}{}
\SetKwProg{InSynchroRound}{In synchronous round}{ do}{}
\SetKwProg{Fn}{Function}{ is}{end~function}
\SetKwProg{IsDefinedAs}{}{\textnormal{ $\isdefinedas$ }}{}
\SetKwIF{If}{ElseIf}{Else}{if}{then}{else if}{else}{}

\SetKwComment{Comment}{$\rhd$}{}
\SetCommentSty{textit}

\newcommand{\commStep}{%
  \medskip%  
  {\SetKwComment{Comment}{}{}\Comment{\fbox{Communication step}}}
  \smallskip%
}

\newcommand{\compStep}{%
  \medskip%  
  {\SetKwComment{Comment}{}{}\Comment{\fbox{Computation step} \hspace{2.3cm} $\rhd$ {\it wait that all messages $\msgm(-)$ of round $R$ are in $\received_{i,R}$}}}
  \smallskip%
}

\SetAlCapSkip{1em}

\usepackage{amsmath,amssymb,amsthm}
\usepackage{thmtools,thm-restate}

\newtheorem{theorem}{Theorem}
\newtheorem{lemma}{Lemma}
\newtheorem{corollary}{Corollary}
\newtheorem{sublemma}{Sublemma}[lemma]

\DeclareMathOperator{\Exists}{\exists\,}

\newenvironment{proofoverview}{
  \renewcommand{\proofname}{Proof overview}%
  \begin{proof}(Detailed proof below.)
}{%
  \end{proof}
}%

\newenvironment{detailedproof}{
  \renewcommand{\proofname}{Detailed proof}%
  \begin{proof}%
}{%
  \end{proof}
}%

\newcommand{\brb}{brb}

\newcommand{\send}{\ensuremath{\mathsf{send}}\xspace}
\newcommand{\broadcast}{\ensuremath{\mathsf{broadcast}}\xspace}

\newcommand{\brbbroadcast}{\ensuremath{\mathsf{\brb\_broadcast}}\xspace}
\newcommand{\brbdeliver}{\ensuremath{\mathsf{\brb\_deliver}}\xspace}

\newcommand{\minfn}{\ensuremath{\mathsf{min}}\xspace}
\newcommand{\maxfn}{\ensuremath{\mathsf{max}}\xspace}

\newcommand{\view}{\ensuremath{\mathit{view}}\xspace}

\newcommand{\received}{\ensuremath{\mathit{received}}\xspace}

\newcommand{\toBeBcast}{\ensuremath{\mathit{to\_bcast}}\xspace}

\newcommand{\orig}{\ensuremath{\mathrm{sender}}\xspace}
\newcommand{\psender}{\ensuremath{p_{\orig}}\xspace}
\newcommand{\wgood}{\ensuremath{w_{\sf good}}\xspace}
\newcommand{\lambdagood}{\ensuremath{\lambda_{\sf good}}\xspace}

\newcommand{\ready}{\ensuremath{\mathit{delivered}}}
\newcommand{\weights}{\ensuremath{\mathit{weights}}}

\newcommand{\stopKW}{\ensuremath{\mathsf{ quit}()}\xspace}

\newcommand{\messageOf}{\ensuremath{\mathsf{message}}\xspace}
\newcommand{\subchain}{\ensuremath{\mathsf{subchain}}\xspace}
\newcommand{\truncate}{\ensuremath{\mathsf{truncate}}\xspace}
\newcommand{\choice}{\ensuremath{\mathsf{choice}}\xspace}
\newcommand{\set}{\ensuremath{\mathsf{set}}\xspace}

\newcommand{\WBPWeightMonoProp}{{\sc WBP-Weight-Monotony}\xspace}
\newcommand{\WBPViewMonoProp}{{\sc WBP-View-Monotony}\xspace}
\newcommand{\WBPRevealWeightMonoProp}{{\sc WBP-Revealing-Round-Monotony}\xspace}
\newcommand{\WBPLocalConspProp}{{\sc WBP-Local-Conspicuity}\xspace}

\newcommand{\WBPConspiProp}{{\sc WBP-Conspicuity}\xspace}
\newcommand{\WBPVisiProp}{{\sc WBP-Final-Visibility}\xspace}
\newcommand{\WBPLiveProp}{{\sc WBP-Good-Case-Liveness}\xspace}

\newcommand{\delpredicate}{\ensuremath{\mathsf{WBP}}\xspace}
\newcommand{\reveal}{\ensuremath{\mathsf{reveal\_round}}\xspace}

\newcommand{\GCLpredicate}{\ensuremath{{\delpredicate_{\mathsf{GCL}}}}\xspace}
\newcommand{\certificate}{\GCLpredicate}
\newcommand{\revealGCL}{\ensuremath{\reveal_{\mathsf{GCL}}}\xspace}

\newcommand{\knownmessages}{\ensuremath{\mathit{known\_msgs}}\xspace}

\newcommand{\setchains}{\ensuremath{\mathit{chains}}\xspace}
\newcommand{\candidatemsg}{\ensuremath{\mathit{candidate\_msgs}}\xspace}

\newcommand{\rmwi}{\ensuremath{r_{i}}\xspace}

\newcommand{\seq}{\ensuremath{\gamma}\xspace}

\newcommand{\sequencei}{\ensuremath{\seq_{i}}\xspace}

\newcommand{\accumulatorTwo}[1]{\ensuremath{T_{2,#1}}\xspace}

\newcommand{\finalround}{\ensuremath{t+1}\xspace}

\newcommand{\BRB}{\MakeUppercase{\brb}}

\newcommand{\brbVprop}{\textsc{\BRB-Validity}\xspace}
\newcommand{\brbNDNprop}{\textsc{\BRB-No-duplication}\xspace}
\newcommand{\brbNDYprop}{\textsc{\BRB-No-duplicity}\xspace}
\newcommand{\brbLDprop}{\textsc{\BRB-Local-delivery}\xspace}
\newcommand{\brbGDprop}{\textsc{\BRB-Global-delivery}\xspace}

\newcommand{\annote}[3]{{\color{#3}%
	\colorbox{#3}{\bfseries\sffamily\tiny\textcolor{white}{#2}}
	\color{#3}
	$\blacktriangleright${\em #1}$\blacktriangleleft$}
}

\newcommand{\ft}[1]{\annote{#1}{FT}{magenta}}
\newcommand{\ta}[1]{\annote{#1}{TA}{blue}}

\renewcommand{\annote}[3]{}

\newcommand{\msgm}{\textsc{msg}\xspace}

\newcommand{\ffalse}{\ensuremath{\mathtt{false}}\xspace}
\newcommand{\ttrue}{\ensuremath{\mathtt{true}}\xspace}

\DeclareMathSymbol{\myColon}{\mathord}{operators}{"3A}
\newcommand{\cryptchain}{{\myColon}}

\newcommand{\isdefinedas}{\stackrel{\text{\tiny def}}{=}}

\title{\textbf{Good-case Early-Stopping Latency of Synchronous\\Byzantine Reliable Broadcast: The Deterministic Case (Extended Version)}}

\author{Timoth\'e Albouy,
  ~Davide Frey,
  ~Michel Raynal,
  ~Fran\c{c}ois Ta\"{i}ani~\\~\\
Univ Rennes, IRISA, CNRS, Inria, 35042 Rennes, France}

\begin{document}

\maketitle

\begin{abstract}
This paper considers the good-case latency of Byzantine Reliable Broadcast (BRB), i.e., the time taken by correct processes to deliver a message when the initial sender is correct. This time plays a crucial role in the performance of practical distributed systems.
Although significant strides have been made in recent years on this question, progress has mainly focused on either asynchronous or randomized algorithms.
By contrast, the good-case latency of deterministic synchronous BRB under a majority of Byzantine faults has been little studied.
In particular, it was not known whether a good-case latency below the worst-case bound of $t+1$ rounds could be obtained.
This work answers this open question positively and proposes a deterministic synchronous Byzantine reliable broadcast that achieves a good-case latency of $\maxfn(2,t+3-c)$ rounds, where $t$ is the upper bound on the number of Byzantine processes and $c$ the number of effectively correct processes.

~
\end{abstract}

\newcommand\nnfootnote[1]{%
  \begin{NoHyper}
  \renewcommand\thefootnote{}\footnote{#1}%
  \addtocounter{footnote}{-1}%
  \end{NoHyper}
}

\nnfootnote{A first version of this paper was initially published at the 36th Int. Symposium on Distributed Computing (DISC 2022).}
\section{Introduction}

%\paragraph{Byzantine reliable broadcast}
Introduced in the eighties~\cite{DS83,LSP82}, Byzantine reliable broadcast (BRB) and Byzantine Broadcast (BB) are two fundamental abstractions of distributed computing~\cite{AW04,AFRT21,B87,CGL11,IR16,MHR14,NRSVX20,WXDS20,R18}. %, which has been extensively studied in the past, both in a asynchronous and synchronous settings.
BRB assumes that one particular process, the sender, broadcasts a message to the rest of the system and that correct (a.k.a. honest) processes all deliver the value initially broadcast if the sender is correct or that, if it is not, either all agree on some value or none delivers any value. BB further requires that all correct processes always deliver some value.\footnote{In this paper, we will tend to conflate the two problems, as the protocols we discuss solve both BB and BRB.}
BRB and BB play a crucial role in many practical distributed applications, from state machine replication (SMR) (see, for instance, the discussion in~\cite{ANRX21}), to broadcast-based money transfer~\cite{AFRT20,BDS20,CGKKMPPST20,GKMPS22}.

\paragraph{Good-case latency}  In  broadcast-based money transfer algorithms, for instance, a cryptocurrency is implemented by merely broadcasting the transfer operations originating from one participant (or in some sharded versions~\cite{BDS20} from one \emph{authority}) to the rest of the system~\cite{AFRT20,CGKKMPPST20,DBLP:conf/pact/FreyGRT21}.
These algorithms do not require consensus, and their performance is directly related to the underlying (Byzantine-tolerant) reliable broadcast algorithm they use.
Transfers issued by correct participants are guaranteed to terminate and only involve a single broadcast operation invoked by the issuer.
As a result, the latency of these algorithms---as experienced by correct participants---solely depends on the \emph{good-case latency} of the BRB algorithm they use, defined as the time taken for all correct parties to deliver a broadcast message when the initial broadcaster is correct~\cite{ANRX21}.
The \emph{good-case latency} of Byzantine-tolerant broadcast algorithms plays a similarly central role in the performance of SMR algorithms, with vast practical consequences for the performance of BFT replication systems, including consortium~\cite{AMNRY20,GKQV10} and committee-based blockchains~\cite{CM19}.

\paragraph{Synchronous networks}  In this paper, we focus on the \emph{good-case latency} of BRB algorithms subject to an \emph{arbitrary number of Byzantine failures} (i.e., we assume $n>t$, where $n$ is the number of processes, and $t$ is an upper bound on the number of Byzantine processes).
We further assume that processes can use signatures to authenticate messages.
We follow in this respect \cite{FN09} and~\cite{WXSD20}, and in part~\cite{ANRX21}.
% As in~\cite{FN09,WXSD20}, and partially in~\cite{ANRX21}, we do not constrain the number of Byzantine failures (i.e. we assume $n>t$, where $n$ is the size of the system, and $t$ is an upper bound on the number of Byzantine processes.).
Since BRB cannot be solved even in a partially synchronous model when $t \geq n/3$~\cite{DLS88,LSP82,PSL80}, we also assume a \emph{synchronous} network, in which messages are delivered during the same round in which they are sent.
Although synchronous wide-area networks are challenging to realize in practice, they can be approximated with high probability by using sufficiently high timeouts.
Synchronous algorithms are further intriguing in their own right and can yield insights into the nature of distributed computing that are useful beyond their specific use.

\paragraph{Randomized synchronous BRB algorithms}  The study of randomized synchronous BRB and BB algorithms tolerating arbitrary many Byzantine faults has progressed substantially in recent years~\cite{ANRX21,FN09,WXSD20}. In particular, the solution proposed by Wan, Xiao, Shi, and Devadas~\cite{WXSD20} and optimized by Abraham, Nayak, Ren, and Xiang~\cite{ANRX21} presents sublinear worst- and good-case latency bounds in expectation (boiling down to constant numbers of rounds when $t$, the maximal number of Byzantine processes, is assumed to be a fraction of $n$).
However, these works all rely on \emph{randomization}.\footnote{In addition, these randomized solutions generally assume a \emph{weakly adaptive adversary}, an adversary that cannot erase messages sent ``just before'' a process becomes Byzantine, where ``just before'' means in the same round. A notable exception is the solution presented in~\cite{WXDS20}, which tolerates a strongly-adaptive adversary by exploiting time-lock puzzles. By contrast, a deterministic algorithm that tolerates Byzantine processes inherently tolerates a strongly-adaptive adversary, i.e. an adversary which can remove messages ``after the fact''.}
Further, these works do not leverage a lower number of actual faults to provide an \emph{early stopping} property~\cite{DRS90}: their latency only involves $n$, the number of processes, and $t$, the upper bound on the number of Byzantine processes, but not $c$, the \emph{effective} number of correct processes. As a result, they cannot exploit a low number of actual failures to provide better latency performance.

\paragraph{This paper's contribution}
In contrast to randomized solutions, the good-case latency of \emph{deterministic} synchronous BRB and BB algorithms has been little studied.
%% By definition, a deterministic Byzantine-tolerant broadcast algorithm tolerates a strongly adaptive adversary (one that can remove messages ``after the fact'').
In the worst case, however, its latency is lower-bounded by $t+1$ rounds~\cite{DS83,DLS88}, and optimal algorithms in this respect have been known since the eighties~\cite{DS83,LSP82}.

An unsolved question to this date is thus whether a good-case latency lower than $t+1$ rounds can be achieved using a deterministic algorithm subject to an arbitrary number of Byzantine faults.
This paper answers this question positively and proposes a deterministic synchronous Byzantine reliable broadcast that achieves a good-case latency of $\maxfn(2,t+3-c)$ rounds, where $t$ is the upper-bound on the number of Byzantine processes, and $c$ the number of effectively correct processes ($c\geq n-t$).
The algorithm we propose does not require correct processes to know either $n$ or $c$.
Moreover, and differently from recently proposed solutions to this problem~\cite{ANRX21,FN09,WXSD20}, our solution:
\begin{itemize}
\item is deterministic, % (which is why it trivially tolerates a strongly adaptive adversary),
\item only relies on signatures, eschewing richer cryptographic primitives (e.g. distributed random coins~\cite{FN09,WXSD20}, verifiable random functions~\cite{MRV99,WXSD20} or time-lock puzzles~\cite{WXDS20}),
\item ensures delivery in just $2$ rounds in good cases as soon as the effective number of  correct processes, $c$, is at least $t+1$, thus improving on all existing solutions.
\end{itemize}

More generally, our good-case latency is \emph{early stopping}~\cite{DRS90}, in that, in good cases, our algorithm will stop earlier when the effective number of correct processes $c$ increases. This provides a substantial advantage even when $c < t+1$. For instance, assuming $t< 3/4\times n$, and an intermediate situation where only $\lfloor t/2 \rfloor$ processes are effectively Byzantine, the good-case latency of our algorithm outperforms that of the best-known randomized algorithm up to $n\leq 43$, and is at least as good up to $n\leq 51$, making it competitive in a wide range of small- to medium-scale practical distributed systems.

% Our algorithm, although not trivial, remains simple. %It exploits patterns in signature chains, thus extending an idea as old as the problem itself~\cite{DS83,LSP82}.
\ft{FT05Jan23: Added to announce generic construction}%
To construct our solution, we introduce a general family of predicates used to select messages that we have termed \emph{weight-based predicates}.
These predicates exploit patterns in signature chains to help correct processes decide when they can safely deliver a message, thus extending an idea as old as the problem itself~\cite{DS83,LSP82}.
We formally define this predicate family, and present a \emph{generic} BRB algorithm that exploits its properties.
We then show how Lamport, Shostak and Pease's seminal BRB algorithm~\cite{LSP82} can be re-interpreted as a specific instance of our generic construction, and finally present our novel solution as a more advanced example with stronger properties, which yield our much-improved good-case latency.

% The reason our algorithm outperforms the best existing solutions in failure free runs is because for a given $n$ and $t$, its round complexity decreases when the number of effectively correct process ($c$) increases.
% Its early stopping nature also stands in contrast to the algorithms proposed recently~\cite{ANRX21,WXSD20,FN09} that use terminations conditions that only depends on $n$ and the upper bound on the number of Byzantine processes $t$ (and do not involve $c$).
%
\section{Background and Related Work}

\sloppy The Synchronous Byzantine Reliable Broadcast problem was first introduced in~\cite{PSL80} 
by Lamport, Shostak, and Pease, who proposed in~\cite{LSP82}
a deterministic solution based on signature chains. This solution requires $t+1$ rounds (both in good and bad cases), where $t<n$ is an upper bound on the number of Byzantine processes present in the system.
This worst-case round complexity was shown by Dolev and Strong~\cite{DS83} to be optimal for deterministic algorithms.
%In a synchronous system, Byzantine reliable broadcast is closely related to Byzantine Consensus (in which all processes must agree on the same output value~\addRef{}.
This result was later refined by Dolev, Reischuk, and Strong who showed that $\minfn(n - 1, n-c + 2, t + 1)$ rounds are necessary to realize Synchronous Byzantine Broadcast~\cite{DRS90}, where  $c\geq n-t$ is the effective number of correct processes in a given run.
They also present in the same paper a deterministic signature-free algorithm that achieves this bound provided that $n > \maxfn(4t, 2t^2 - 2t + 2)$.
The salient properties of this algorithm are summarized in the first column of Table~\ref{tab:comparison:sota} and compared to more recent works and to this paper (last column).

% , and generally assuming a weakly adaptive adversary, i.e., an adversary that can adaptively corrupt processes, but cannot remove messages sent in the round when a process becomes Byzantine
In recent years, substantial progress has been achieved to circumvent the hard bound of $t+1$ rounds for deterministic BRB and BB algorithm by exploiting \emph{randomization}. Assuming a majority of Byzantine processes, Fitzi and Nielsen proposed in~\cite{FN09} a randomized algorithm that achieves Byzantine Agreement in an expected number of $\lfloor(3t-n)/2\rfloor + 7 + O(1)$ rounds\footnote{More precisely, this expected number of rounds can be broken down into a deterministic number of synchronous rounds followed by an expected number of asynchronous rounds. The exact breakdown depends, in turn, on the choice of shared random coin used in the algorithm.}, and a good-case latency of $\lfloor (3t- n)/2 \rfloor + 6$ deterministic rounds.

\begin{table}[tb]
  \centering
  \renewcommand{\tabularxcolumn}[1]{>{\small}m{#1}}
  {\small
  \newcommand{\yes}{\bf yes}
  \newcommand{\myskip}{0.05em}
  \setlength{\tabcolsep}{1.5pt}
  \begin{tabularx}{\linewidth}{|>{\footnotesize}X|c|c|c|c|}
    \hline
    &Dolev, Reischuk&& \footnotesize Wan et al.~\cite{WXSD20} + &\\
    & %\raisebox{0.5em}[0pt][0pt]{Dolev et al.~%% Reischuk \& Strong~
    \& Strong~
    \cite{DRS90} & 
    \raisebox{0.5em}[0pt][0pt]{Fitzi \& Nielsen
    ~\cite{FN09}} &
    %    Wan, Xiao, Shi, \& Devadas~\cite{WXSD20}, optimized by Abraham, Nayak, Ren, \& Xiang~\cite{ANRX21}
    \footnotesize Abraham et al.~\cite{ANRX21} &
    \raisebox{0.5em}[0pt][0pt]{This paper}\\\hline
    Deterministic & \yes & no & no & \yes\\[\myskip]
%    Strong adversary & \yes & no & no & \yes\\[\myskip]
    Early stopping & \yes & no & no & \yes\\[\myskip]
    Dishonest majority & no & \yes & \yes & \yes\\[\myskip]
    $n>$& $\maxfn(4t,2t^2{-}2t{+}2)$ & $-$ & $-$ & $-$\\[0.3em]
    Worst-case latency& $\minfn(n{-}c{+}2,t{+}1)$ &
    % $\begin{aligned}\textstyle\maxfn(7,\lfloor\frac{3t-n}{2}\rfloor{+}7)\\     {+}O(1)\,^{\ast}\end{aligned}$
    $\maxfn(7,\lfloor\frac{3t-n}{2}\rfloor{+}7){+}O(1)\,^{\ast}$
    &
    $O\big( (\frac{n}{n-t})^2 \big)\,^{\ast}$&
    $t+1$\\[0.3em]
    Good-case latency& 2 & $\maxfn(6,\lfloor\frac{3t-n}{2}\rfloor+6)$ &
    $\big\lceil \frac{n}{n-t}\big\rceil + \big\lfloor\frac{n}{n-t}\big\rfloor$ & $\maxfn(2,t{+}3{-}c)$  \\[0.3em]
    \hline
  \end{tabularx}}
\caption{Assumptions, guarantees, and latencies of synchronous signature-based BRB algorithms
($^\ast$ indicates an expected number of rounds)
}
  
 % A comparison of assumptions and guarantees of previous synchronous BRB algorithms with signatures. An asterisk ($^\ast$) indicates an expected number of rounds. $n$ is the number of processes in the system, $t$ an upper bound on the number of Byzantine processes ($t<n$), and $c$ the effective number of correct processes ($c\geq n-t$).}
  \label{tab:comparison:sota}
\end{table}

In 2020, Wan, Xiao, Shi, and Devadas presented a randomized algorithm that achieves BB in $O\big((\frac{n}{n-t})^2\big)$ expected synchronous rounds~\cite{WXSD20}. Last year, in an in-depth study of the good-case latency of BB and BRB algorithms~\cite{ANRX21} (extended version in~\cite{ANRX21-2}), Abraham, Nayak, Ren, and Xiang proved a lower bound of $\lfloor n/(n-t) \rfloor -1$ rounds for the good-case latency of synchronous BRB. They then explained how the solution presented in~\cite{WXSD20} can be optimized to deliver a good-case latency of $\lceil n/(n-t) \rceil + \lfloor n/(n-t)\rfloor$ rounds (about $2n/(n-t)\pm 1$).%% assuming a weakly adaptive adversary.
\footnote{Although correct processes can deliver their message in about $2\times n/(n-t)$ rounds in this optimized algorithm, they must continue to participate in the algorithm for about the same amount of time, leading to an overall execution time of circa $4\times n/(n-t)$ rounds in good-cases.}

The properties of these earlier works are summarized in Table~\ref{tab:comparison:sota}, together with those of the algorithm we propose.
Among these works, only~\cite{DRS90} is deterministic. % and therefore tolerates a strongly adaptive adversary.
It imposes, however, a strong constraint on $n$ ($n > \maxfn(4t,2t^2{-}2t{+}2)$) and does not tolerate a majority of Byzantine processes, which the other algorithms do.
Conversely, the algorithms of~\cite{ANRX21,FN09,WXSD20} all tolerate an arbitrary number of Byzantine processes, but contrary to the solution we present, they rely on randomization %% under a weakly adaptive adversary
and do not exploit executions in which the number of Byzantine processes is less than the upper bound $t$.
(They are not early stopping.)

% Although these earlier works did not explicitly target good-case latency (i.e. cases in which the sender is honest), this aspect has been investigated in depth by Abraham, Nayak, and Ren in a recent publication~\cite{ANRX21}. In particular, Abraham, Nayak, and Ren offer %% Although the conference paper focuses primarily on the asynchronous case, but the accompanying technical report offer 
% key results applying to synchronous BRB. They first prove a lower bound of $\lfloor \frac{n}{n-t} \rfloor -1$ rounds for the good-case latency of synchronous BRB. They then shows how the solution proposed by Wan, Xiao, Shi, and Devadas in~\cite{WXSD20} 
%
\section{Computing Model and Specification}

\subsection{System Model}

\paragraph{Process Model} The system consists of $n$ synchronous sequential processes denoted $\Pi=\{p_1$, ..., $p_n\}$. 
Each process $p_i$ has an identity; all the identities are different and known by all processes. 
To simplify, we assume that $i$ is the identity of $p_i$.

Regarding failures, up to $t$ processes can be Byzantine, where a Byzantine process is a process whose behavior does not follow the code specified by its algorithm~\cite{LSP82,PSL80}. 
Let us notice that Byzantine processes can collude to fool non-Byzantine processes (also called correct processes).
Let us also notice that, in this model, the premature stop (crash)
of a process is a Byzantine failure. 
The integer $c$ denotes the number of processes that effectively behave correctly in an execution.
Both $c$ and $n$ remain unknown to correct processes, but they are used to analyze the properties of our algorithm.

\paragraph{Network Model}
Processes communicate by exchanging messages through a reliable synchronous network, in which messages are delivered in the round in which they were sent. %~\cite{FN09,WXSD20}.

\paragraph{Security Model}
Like earlier works in this area~\cite{DRS90,FN09,LSP82,PSL80,WXSD20}, we assume a PKI (Public Key Infrastructure) that provides an ideal signature scheme.
Processes can sign the messages they send, verify signatures, and forward content signed by other processes.

\subsection{Byzantine Reliable Broadcast}

Following~\cite{ANRX21,FN09,WXSD20}, we consider a one-shot Byzantine-tolerant reliable broadcast (BRB for short) in which the sending process $p_{\orig}$ is known beforehand. % Any algorithm implementing this specification can then be easily generalized to a multi-sender multi-shot broadcast by instantiating different instances for each message and each sender.
The BRB abstraction provides two operations, \brbbroadcast and \brbdeliver.
$\brbbroadcast(m)$ is invoked by the sending process $p_{\orig}$.
When this happens, we say that $p_{\orig}$ \brb-broadcasts $m$.
When a process $p_i$ invokes $\brbdeliver(m)$ we say that $p_i$ \brb-delivers $m$.
The BRB abstraction is specified by the following five properties.

\begin{itemize}
\item Safety:
  \begin{itemize}
  \item \brbVprop:
    If a correct process $p_i$ \brb-delivers a message $m$ %% (from $p_{\orig}$,
    and $p_{\orig}$ is correct, %% with sequence number \sn
    then $p_{\orig}$ has \brb-broadcast $m$. % with sequence number~\sn.
  \item \brbNDNprop:
    A correct process $p_i$ \brb-delivers at most one message. % from a process $p_j$. % with sequence number~\sn.
  \item \brbNDYprop:
    No two different correct processes \brb-deliver different messages. % from a process $p_i$. % with the same sequence number~\sn.
  \end{itemize}
\item Liveness:
  \begin{itemize}
  \item \brbLDprop:
    If $p_{\orig}$ is correct and \brb-broadcasts
    a message, then at least one
    correct process $p_j$ eventually \brb-delivers some message.
  \item \brbGDprop:
    If a correct process $p_i$ \brb-delivers a message, then all correct processes \brb-deliver a message.
  \end{itemize}
\end{itemize}
  
%% \begin{proof}
%% By contradiction, let $p_j$ and $p_k$ be two correct processes that respectively \brb-deliver $m$ and $m'$ (with $m \neq m'$) with sequence number~\sn from a process $p_i$.

%Si p_i et p_i' \brb-délivrent dans 2 rondes différentes : Mettons que p_i \brb-délivre en premier. Avant de \brb-délivrer, p_i dissémine MSG(m,sn,j,-), et p_i' le reçoit à la ronde d'après. Etant donné que p_i' \brb-délivre quelque chose, alors il n'a pas reçu d'autres messages MSG(-,sn,j,-) en conflit entre temps, et donc il \brb-délivre bien (m,sn,j).
%Si p_i et p_i' \brb-délivrent à la même ronde : p_i et p_i' ont respectivement disséminé MSG(m,sn,j,-) et MSG(m',sn,j,-) puis attendu (au moins) une ronde avant de \brb-délivrer, donc ils ont tout les deux reçu le message de l'autre. Etant donné qu'ils \brb-délivrent quelque chose, alors il leurs messages ne sont pas en conflit, et donc ils \brb-délivrent tous les deux la même chose.
% \end{proof}

%% \subsection{Time complexity}

%% If sender is correct: at round 1 sender sends, at round $2$, all correct receive 1 signature and send own signature, so at round $3$, all correct have at least c signatures, so they wait max($t+2-c,0$) rounds starting from round 3. 
%% If $c \geq t+2$, so at round $3$, all correct have at least $c=t+2$ signatures, so they wait $0$ rounds, and deliver immediately. So the algorithm takes $2$ rounds. 
%% If we have $c=t+1$, it takes $3$ rounds, and so on.  If send is byzantine, it's $3$ (or $2$, or $4$) starting from the round at which the first correct learns of the message. 
%
%% -*- coding: utf-8 -*-

\section{A Generic BRB Algorithm Based on Weight-based Predicates}
\label{sec:fram-constr-determ}

The BRB algorithm introduced in this paper exploits patterns in sets of signature chains to detect when a (correct) process can safely \brb-deliver a message $m$ earlier than the worst-case latency $t+1$. We introduce this new algorithm in two steps: in this section, we first present a family of predicates used to select and rank messages and a generic BRB algorithm based on this predicate family that is designed to provide interesting good-case latency values. In the next section, we then define a particular predicate belonging to this family that achieves a good-case latency of $\maxfn(2,t+3-c)$ rounds when used in the generic algorithm of this section.

The family of predicates we introduce exploits \emph{weights} to create a hierarchy between potential candidate messages, and relies on two central properties, \emph{conspicuity} and \emph{final visibility}, that allow correct processes to \brb-deliver a message early yet safely in favorable circumstances.
In the following, we define the properties predicates belonging to this family must fulfill, show that these properties are sufficient to implement BRB with our generic algorithm (Theorem~\ref{theo:algo:generic:SBRB:pi:works}), and finally illustrate how Lamport, Shostak, and Pease's (BRB-LSP in the following) seminal algorithm can be seen as a specific example of our generic construction.

\subsection{Underlying intuition}

\paragraph{Signature chains} 
% As already said, Lamport, Shostak, and Pease were the first to describe the Byzantine Reliable Broadcast abstraction and to propose an algorithm to implement it in a synchronous network with signatures.
The original BRB-LSP algorithm uses \emph{signature chains} to propagate what each process knows of the system's state~\cite{LSP82}.
A signature chain (or chain for short) starts with a message $m$ signed by the sending process, e.g. $(m,i_\orig,\sigma_{\psender})$, where $i_\orig$ is the identity of the sending process, and $\sigma_{\psender}$ is a signature of $(m,i_\orig)$ with $\psender$'s private key.
Such a chain is of length $1$, as it contains one signature. A chain of length $\ell$ is extended by appending the identity $i_{\ell+1}$ of a process $p_{i_{\ell+1}}$ not present in the chain, followed by $p_{i_{\ell+1}}$'s signature of the resulting sequence:
$$(m,i_\orig,\sigma_{\psender},i_2,\sigma_{p_{i_2}},..,i_\ell,\sigma_{p_{i_\ell}},i_{\ell+1},\sigma_{p_{i_{\ell+1}}}).$$

\noindent As in~\cite{DS83,LSP82}, we use the compact notation $m\cryptchain \psender \cryptchain p_{i_2} \cryptchain .. \cryptchain p_{i_{\ell+1}}$ to represent such a chain.

\paragraph{Valid chains} In BRB-LSP~\cite{LSP82}, further formalized in~\cite{DS83}, and algorithms based on the same idea~\cite{FN09}, correct processes only accept \emph{valid} signature chains, i.e., signature chains that are acyclic and whose length matches the current round.
% In round $R$, a signature chain is valid if (i) all its signature are valid, (ii) it contains no cycles (each process only appear at most once), and (iii) it is of length $R$.
% In 's original algorithm~\cite{LSP82}, only signature chains\todo{revise to explain what is a chain} that are acyclic and whose length matches the current round are taken into account.
These conditions constrain the disruption power of Byzantine processes by limiting how long they can hide a message from correct processes. 
In~\cite{DS83,LSP82}, a message is considered for \brb-delivery when backed by at least one chain containing $t+1$ signatures: the length of the chain ($t+1$) ensures that Byzantine processes cannot reveal some message $m$ to only a subset of correct processes, while hiding it from others, and thus guarantees that all correct processes use the same set of messages to decide which message should be \brb-delivered (using a deterministic choice function).

\paragraph{From chains to weight-based predicates}
The generic algorithm proposed in this paper generalizes this intuition in a simple, albeit non-trivial, way.
Instead of single chains, our algorithm detects \emph{sets of chains} forming a particular \emph{pattern} to trigger delivery.
Ideally, such a pattern should allow correct processes to \brb-deliver early in good cases, while remaining safe in bad ones.
To fulfill this goal, our algorithm adopts the same sign-and-retransmit strategy as BRB-LSP.
This means that, in a good case execution, \psender signs a unique message $m$, and this message is necessarily repeated in round $2$ by the $c-1$ remaining correct processes, totalling at least $c$ signatures ``backing'' $m$ by the end of round 2 (that of \psender in round 1, and the other $c-1$ correct processes in round 2).
Although correct processes do not know $c$ (they only know the lower bound $n-t$), they can thus assign a \emph{weight} to each message $m$ they observe, depending on the ``amount'' of backing this message is perceived to enjoy from other processes.

Our intuition consists in combining this notion of weight with the temporal information provided by synchronous rounds to obtain a safe yet good-case-ready \emph{weight-based predicate}. Such a weight-based predicate serves to select candidate messages for delivery, while the perceived weight of a message serves to discriminate between competing messages when \psender is Byzantine. The crux of the problem lies in ensuring that a message $m$ selected by this predicate and \brb-delivered early by a correct process $p$ will trump (thanks to its weight) any other potential competitor $m'$ that might surface in later rounds. We solve this difficulty by requiring two properties from a weight-based predicate:
\begin{itemize}
    \item A weight-based predicate should be \emph{conspicuous}, in the sense that if a process $p$ observes a predicate of weight $w$ for a message $m$, $m$ should necessarily have become visible to all correct processes at or before \emph{a specific revealing round} that only depends on $w$, and $n$, $t$.
    By contrapositive, this property allows correct processes to conclude to the \emph{nonexistence} of a predicate of weight $w$ for a message $m'$ if they have not heard of $m'$ after this revealing round, and is thus instrumental to determining that a message $m$ cannot be beaten by any other.
    
    \item A weight-based predicate should also be \emph{finally visible}, meaning that if $p$ perceives a predicate of weight $w$ for $m$, all other correct processes should also observe  a predicate of weight $w$ for $m$ at the latest by round $t+1$.
\end{itemize}

In good cases, the conspicuity and final visibility of a weight-based predicate make it possible for a process $p$ that envisages to \brb-deliver a message $m$ to know that (i) no other message can beat $m$ in terms of weight (by waiting until the corresponding revealing round for $m$'s weight), and (ii) that all correct processes will also assign a weight of $w$ to $m$ at the latest by round $t+1$ (thus ensuring that they will also \brb-deliver $m$).
In bad cases, conspicuity prevents Byzantine processes from tricking a correct process into delivering early while revealing contradictory information to other correct processes in later rounds, and final visibility guarantees that correct processes observe the same (weight,message) pairs in the final $t+1$ round, ensuring agreement.

\paragraph{The rest of this section} The remainder of this section, first introduces a few notations that we use to manipulate (sets of) signature chains and messages (Section~\ref{sec:notations}).
The generic synchronous BRB algorithm (Algorithms~\ref{alg:generic:porig:code} and~\ref{algo:generic:SBRB:pi}) and the weight-based predicates it relies on are described in Section~\ref{sec:generic-weight-based}.
We then formalize the properties that weight-based predicates and their revealing functions must fulfill for our generic algorithm to implement BRB (Section~\ref{sec:lambd-brb-robustn}), a connection that is captured by Theorem~\ref{theo:algo:generic:SBRB:pi:works} (Section~\ref{sec:proof-theor-refth}).
Finally, to illustrate the generality of the proposed algorithm, we show that BRB-LSP can be interpreted as a special case of our generic construction (Section~\ref{sec:revis-lamp-shost}).

%% We call these chain patterns \emph{weight-based predicates}.

%% The heavier a certificate, the more quickly a correct process can make a delivery decision in the absence of any contradictory information.
%% This approach is beneficial when the initial sender is correct, allowing correct processes to terminate in this case in $\maxfn(2,t+3-c)$ rounds\footnote{\label{fn:comment:on:crashed:proc}Similarly to other works studying synchronous broadcast with a dishonest majority~\cite{ANRX21,FN09,WXSD20}, the presented algorithm considers crashed processes as Byzantine, providing no guarantees for them.
%% A simple change,  however, which adds one extra round, can ensure that crashed processes that \brb-delivers benefit from the \brbNDYprop and \brbGDprop properties. See footnote~\ref{fn:more:detail:on:crashed:proc}.
%% %This change adds one additional round, and yields a good-case complexity of $\maxfn(3,t+4-c)$ rounds.
%% }.

%%%%%%%%%%%%%%%%%%%%%%%%%%%%%%%%%%%%%%%%%%%%%%%%%%%%%%%%
\subsection{Notations}
\label{sec:notations}

\newcommand{\setofchains}{E}

We use the following notations:
\begin{itemize}
\item $m\cryptchain p_{i_1} \cryptchain p_{i_2} \cryptchain \dotsb \cryptchain p_{i_\ell}$ is a chain of signatures (or \emph{chain} for short) as in \cite{DS83,FN09,LSP82}. We say that the \emph{length} of the chain is $\ell$.
  A \emph{valid} chain must start with $p_{\orig}$ (i.e. $p_{\orig}=p_{i_1}$), only contain valid signatures, and be acyclic (a process' signature can only appear once in a given chain).
  As in~\cite{DS83}, we assume a filter function removes any invalid chain from the reception queue of correct processes, so that correct processes only receive valid chains.
  In particular, correct processes will only accept chains of length $R$ during round $R$. As a shortcut, we might therefore say that a process $p_i$ \textit{has signed a chain $\pi$ in round $R$} to mean that $p_i$'s signature is the $R^\mathrm{th}$ signature in $\pi$.
  
\item $\pi$ being a chain of signatures, $\messageOf(\pi)$ 
denotes the message at the start of the chain. We therefore have $\messageOf(m\cryptchain p_{\orig} \cryptchain p_{i_2} \cryptchain \dotsb \cryptchain p_{i_\ell})=m$.
By extension, if $\setofchains$ is a set of chains, $\messageOf(\setofchains)$ is the direct image of $\setofchains$ by $\messageOf()$.

\item $M$ being a set of messages, $\choice(M)$ deterministically returns one of the messages, i.e., the same message $m$ is returned by all correct processes for the same input set $M$.
The function $\choice()$ can be implemented in various ways (e.g., the message with the smallest value or smallest time-stamp).
If $M$ is empty, $\choice(M)$ returns $\bot$.

\newcommand{\indexseq}[1]{i_{#1}}
\item  $\seq=(p_{\indexseq{k}})_{k\in[1..\ell]}\in \Pi^\ell$ 
  being a sequence of $\ell$ processes, for simplicity, we use the notation $\cryptchain\seq\cryptchain$ as a shorthand for the fragment of signature chain $\cryptchain p_{\indexseq{1}}\cryptchain\dotsb\cryptchain p_{\indexseq{\ell}}\cryptchain$.
  For instance, $m\cryptchain p_{\orig}\cryptchain \seq\cryptchain p_i$ thus means $m\cryptchain p_{\orig}\cryptchain p_{\indexseq{1}}\cryptchain\dotsb\cryptchain p_{\indexseq{\ell}} \cryptchain p_i$. We similarly equate the sequence $\seq$ with its supporting set $\set(\seq)=\{p_{\indexseq{k}}\}_{k\in[1..\ell]}$ when unambiguous. Thus $q\in\seq$ means $q\in\{p_{\indexseq{k}}\}_{k\in[1..\ell]}$, $|\seq|=|\{p_{\indexseq{k}}\}_{k\in[1..\ell]}|=\ell$, $X\cup \seq=X\cup \{p_{\indexseq{k}}\}_{k\in[1..\ell]}$.

  For simplicity, we extend these notations to chains of signatures. For instance, if 
  $\pi=m\cryptchain p_{i_1} \cryptchain p_{i_2} \cryptchain \dotsb \cryptchain p_{i_\ell}$ is a chain and $p\in\Pi$ a process, $p\in\pi$ means $p\in\{p_{i_k}\}_{k\in[1..\ell]}$.

\item If $\Gamma\subseteq\Pi^{\ast}$ is a set of process sequences (resp. a set of chains), by abuse of notation we note $\set(\Gamma)$ the set of processes that appear in one of the sequences of $\Gamma$ (resp. whose signature appears in one of the chains of $\Gamma$):
  \begin{equation*}
    \set(\Gamma)=\bigcup_{\gamma\in\Gamma} \set(\gamma).
  \end{equation*}
  
\item  $\seq=(p_{\indexseq{k}})_{k\in[1..\ell]}\in \Pi^\ell$ 
  being a sequence of $\ell$ processes, we note $\subchain(\seq,k_1,k_2)$ the sub-sequence of $\seq$ that contains its $k_1$\textsuperscript{th} to $k_2$\textsuperscript{th} elements (with both $p_{\indexseq{k_1}}$ and $p_{\indexseq{k_2}}$ included). The resulting sub-sequence is truncated accordingly if $\seq$ does not contain enough elements.
  Formally, we have:
  \begin{equation*}
    \subchain(\seq,k_1,k_2) = (p_{\indexseq{k}})_{k\in[k_1..\minfn(\ell,k_2)]}.
  \end{equation*}
  If $|\seq|\leq k$ in particular, $\subchain(\seq,1,k)=\seq$.

  As above, we extend this definition to chains of signatures. If 
  $\pi=m\cryptchain p_{i_1} \cryptchain p_{i_2} \cryptchain \dotsb \cryptchain p_{i_\ell}$ is a chain then $\subchain(\pi,k_1,k_2)=(p_{i_k})_{k\in[k_1...\minfn(\ell,k_2)]}$.
  % $\subchain_k(m\cryptchain p_{\orig} \cryptchain p_{\indexseq{2}} \cryptchain \dotsb \cryptchain p_{\indexseq{\ell}})= m\cryptchain p_{\orig} \cryptchain p_{\indexseq{2}} \cryptchain \dotsb \cryptchain p_{i_k}$.

\item $\pi=m\cryptchain p_{i_1} \cryptchain p_{i_2} \cryptchain \dotsb \cryptchain p_{i_\ell}$ being a chain of signature, we note $\truncate_k(\pi)$ the chain in which only the first $k$ signatures are kept (or all of $\pi$ if $|\pi|\leq k$): $\truncate_k(\pi) = m\cryptchain p_{i_1} \cryptchain p_{i_2} \cryptchain \dotsb \cryptchain p_{i_{\minfn(k,\ell)}}.$
  
\end{itemize}

\subsection{A generic weight-based synchronous BRB algorithm}
\label{sec:generic-weight-based}

\begin{algorithm}[tb]
  \small
  \InSynchroRound{$R=1$}{
    $\broadcast(\msgm(\{m\cryptchain p_{\orig}\}))$\; \label{line:generic:porig:bcast}
    $\brbdeliver(m)$. \label{line:generic:porig:brdeliver}
  }
  \caption{\brb-broadcast operation executed by $p_\orig$ at round $R=1$}
  \label{alg:generic:porig:code}
\end{algorithm}

\begin{algorithm}[tb]
\small
\InitStep{$\view_i \gets \emptyset$;
$\ready_i \gets \ffalse$;
$\toBeBcast_{i,r} \gets \emptyset$ \textbf{for all} $r\in[1..t+1]$%
\label{line:generic:viewi:modif}}{}\medskip

\InEachSynchroRound{$R\in [1..\finalround]$}{\label{line:generic:start:synch:round}

    % #################################### 
    \commStep
    
    \lIf{$R \geq 2$}{\broadcast $\msgm(\toBeBcast_{i,R})$} \label{line:generic:bcast} \label{line:generic:start:comm:step}
    \lIf*{$\ready_i$}{\stopKW.} \label{line:generic:stop:ready} \label{line:generic:end:comm:step}
    
    % ####################################

    \compStep
    
    %$\view_i[R] \gets \text{set of all chains received in round }R$\; \label{line:generic:view:i:R} \label{start:generic:computation:step}
    $\view_i \gets \view_i \cup \{\pi \in \setchains \mid \msgm(\setchains) \in \received_{i,R}\}$\; \label{line:generic:view:i:R} \label{start:generic:computation:step}
    
    $\toBeBcast_{i,R+1} \gets \{\pi \cryptchain p_i \mid \pi \in \view_i[R] \land p_i \not\in \pi \}$\; \label{line:generic:signing:next:round}
    
    %$\knownmessages_{i,R} \gets \text{set of all signed messages in }\view_i$\; \label{line:generic:computing:knownmsg}
    $\knownmessages_{i,R} \gets \{m \mid \Exists m \cryptchain \pi \in \view_i$\}\; \label{line:generic:computing:knownmsg}

    \uIf{$\knownmessages_{i,R} = \{m\} \land \Exists w\in\mathbb{N}^{+}: \ovalbox{$\delpredicate(m,w,\view_i)$} \land R\geq \ovalbox{$\reveal(w)$}$}{ \label{line:generic:test:knowmessage}
      $\brbdeliver(m)$;
      $\ready_i \gets \ttrue$\; \label{line:generic:brdeliver:immediate}
    }
    \ElseIf{$R=\finalround$}{
        \label{line:generic:reached:finalround}
        $\weights_i \gets \{w \in \mathbb{N}^{+} \mid \Exists m \in \knownmessages_{i,R}: \ovalbox{$\delpredicate(m,w,\view_i)$}\}$\;
        \label{line:generic:no:delivery}\label{line:generic:Wi:set}
        
        \uIf{$\weights_i \neq \emptyset$}{ \label{line:generic:cond:weight:not:empty}
            $\candidatemsg_i \gets \{m \in \knownmessages_{i,R} \mid \ovalbox{$\delpredicate\big(m,\maxfn(\weights_i),\view_i\big)$}\}$\label{line:generic:wmaxj:computed}\;
        
            $\brbdeliver(\choice(\candidatemsg_i))$\; \label{line:generic:brdeliver:finalround}
        }
        \lElse*{$\brbdeliver(\bot)$.}\label{line:generic:brbdeliver:bot}
    }
    \label{end:generic:computation:step}
}
\label{line:generic:end:synch:round}
\caption{Skeleton of a Weight-based Synchronous BRB algorithm (code for $p_i\neq p_{\orig}$). The use of $\delpredicate()$ and $\reveal()$ are highlighted using rounded boxes.}
\label{algo:generic:SBRB:pi}
\end{algorithm}

Algorithms~\ref{alg:generic:porig:code} and~\ref{algo:generic:SBRB:pi} describe a general construction for synchronous deterministic broadcast algorithms that lends itself to good-case latency and early-stopping properties. Our approach is modular and hinges on two functions: a message-selection predicate (noted $\delpredicate$, standing for \emph{weight-based predicate}), and a revealing-round function (noted $\reveal$). % The generic code of our construction is shown in Algorithms~\ref{alg:generic:porig:code} and~\ref{algo:generic:SBRB:pi}.

For readability, Algorithm~\ref{alg:generic:porig:code} presents the code for the sending process $\psender$ separately. To \brb-broadcast $m$, $\psender$ simply signs $m$, produces the signature chain $m\cryptchain \psender$, and broadcasts a protocol message $\msgm(\{m\cryptchain \psender\})$ containing this chain to all correct processes before \brb-delivering $m$ locally. Here, and as in the rest of the paper, the operation $\broadcast(m)$ is used as a shorthand for ``{\bf for all $p_j\in\Pi$ do} \send $m$ to $p_j$ {\bf end for}''.

Algorithm~\ref{algo:generic:SBRB:pi} constitutes the core of the generic BRB algorithm we propose.
It  uses up to $t+1$ synchronous rounds (lines~\ref{line:generic:start:synch:round}-\ref{line:generic:end:synch:round}). $R$ is a global variable containing the current round number.
Each round is divided into a communication step (lines~\ref{line:generic:start:comm:step}-\ref{line:generic:end:comm:step}), in which processes broadcast and receive messages exchanged during the round, and a computation step (lines~\ref{start:generic:computation:step}-\ref{end:generic:computation:step}) in which they handle received messages and prepare the messages to be sent during the next round.
The set $\received_{i,R}$ represents the messages received by process $p_i$ during round $R$.
It is directly updated by the (synchronous) network layer.

The set $\toBeBcast_{i,R}$ contains the signature chains to be broadcast by $p_i$ during round $R$.
In the first round, $p_i\neq \psender$ remains silent. %(line~\ref{line:generic:bcast} with $\toBeBcast_{i,1}=\emptyset$).
Process $p_i$ accumulates in the set $\view_i$ the signature chains it receives during each round (line~\ref{line:generic:view:i:R}).
The notation $\view_i[R]$ used at line~\ref{line:generic:signing:next:round} is a shortcut to denote the chains of $\view_i$ that contain exactly $R$ signatures and have therefore just been received.
More generally we use the notation $\view_i[r]\isdefinedas\{\gamma\in\view_i: |\gamma|=r\}$.
The chains of length $R$ that do not already contain $p_i$'s signature are signed by $p_i$ and stored for broadcasting in the next round (line~\ref{line:generic:signing:next:round}).

\sloppy Process $p_i$'s behavior in the computation step is driven by the predicate $\delpredicate$ and the function $\reveal$. 
The function $\delpredicate$ takes three parameters: a message $m$, a positive weight $w\in\mathbb{N}^{+}$, and a set of \emph{valid} chains, $\view$, which captures a process's current view.
It returns a Boolean value, which when true, indicates that $m$ can be considered as a possible message to be \brb-delivered with a weight $w$ according to the information contained in $\view$. More formally this can be expressed as
\begin{equation*}
  \begin{array}{cccc}
  \delpredicate:& \mathcal{M}\times\mathbb{N}^{+}\times\mathcal{P}\big(\mathsf{valid}(\mathcal{M}\times\Pi^{\ast})\big)&\to&
  \{\ttrue,\ffalse\}\\
  &(m,w,\view) &\mapsto& \delpredicate(m,w,\view),
  \end{array}
\end{equation*}
where $\mathcal{M}$ is the set of possible messages, $\mathbb{N}^{+}$ the set of positive integers, and $\mathcal{P}\big(\mathcal{M}\times\mathsf{valid}(\Pi^{\ast})\big)$ the powerset of valid signature chains.
In terms of vocabulary, we say that $p_i$ \emph{observes a predicate} of weight $w$ for a message $m$ during round $R$ if $\delpredicate(m,w,\view_i)=\ttrue$ during the computation step of round $R$ at $p_i$ once the new value of $\view_i$ has been computed (lines~\ref{line:generic:test:knowmessage}-\ref{end:generic:computation:step} of Algorithm~\ref{algo:generic:SBRB:pi}).
The function $\reveal$ is closely linked to $\delpredicate$, and helps a process decide when a message of weight $w$ (according to the predicate $\delpredicate$) is safe to \brb-deliver.
\begin{equation*}
  \begin{array}{cccc}
    \reveal:&\mathbb{N}^{+}&\to&[1..t+1]\\
        &w&\mapsto& \reveal(w).
  \end{array}
\end{equation*}

How $p_i$ uses the information provided by \delpredicate and \reveal depends on whether $p_i$ has reached round $t+1$ or not. In earlier rounds, $p_i$ uses the conspicuity property of the predicate $\delpredicate$ to detect if a message $m$ is backed by a predicate ``heavy enough'' that cannot be beaten by any other message $m'\neq m$ (condition at line~\ref{line:generic:test:knowmessage}).
If this is the case, $m$ is \brb-delivered at line~\ref{line:generic:brdeliver:immediate}, and the flag $\ready_i$ is toggled to stop the algorithm in the next round.\footnote{\label{fn:more:detail:on:crashed:proc}
The extra round of communication induced by $\ready_i$ is needed to ensure all correct processes observe the same predicate $\delpredicate$ as $p_i$.
However, by delivering as soon as the condition of line~\ref{line:generic:test:knowmessage} is true, the algorithm does not ensure that crashed processes benefit from the \brbNDYprop and \brbGDprop properties.
These additional guarantees can be provided at the cost of one extra round by postponing the \brb-delivery of $m$ by one round from line~\ref{line:generic:brdeliver:immediate} to line~\ref{line:generic:stop:ready}.
% See footnote~\ref{fn:comment:on:crashed:proc}.
}
``Heavy enough'' means that $w$, the weight of the predicate observed by $p_i$, should have revealing round $\reveal(w)$ of at most $R$.
This implies that, by round $R$, all predicates of weight $w$ or more must have become conspicuous and allows $p_i$ to make a safe \brb-delivery because it knows that no message $m'$ can exhibit a predicate heavier or equal to $w$, ensuring that $m$ will prevail in case of conflicts possibly detected by other correct processes.
Figure~\ref{fig:predicate-propagation} illustrates the above mechanism.

If $p_i$ reaches round $t+1$ without having \brb-delivered any message (line~\ref{line:generic:reached:finalround}), it tallies all messages known to it and keeps only messages backed by a predicate $\delpredicate$ with maximal weight $\maxfn(\weights_i)$.
Process $p_i$ uses a deterministic function $\choice()$ to break any tie between messages.
Finally, if no message is known to $p_i$ ($\weights_i=\emptyset$), $p_i$ \brb-delivers a default special value, $\bot$ (line~\ref{line:generic:brbdeliver:bot}).

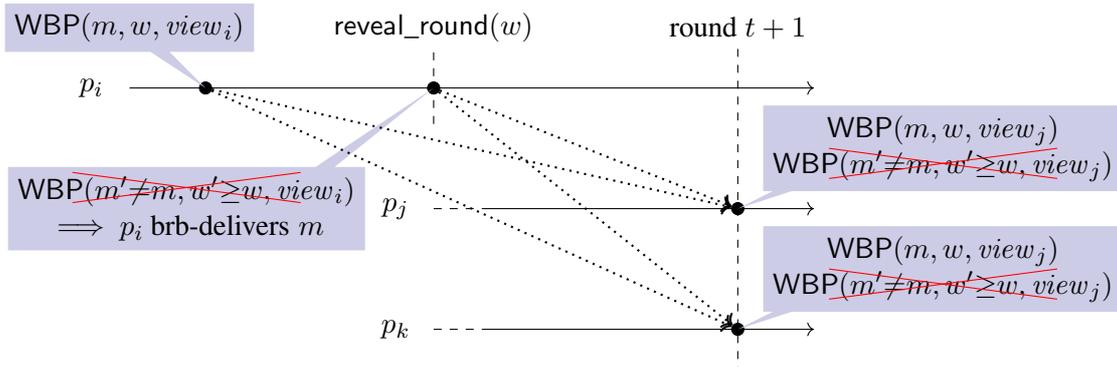
\begin{figure}[ht]
  \centering
\begin{tikzpicture}[x=1cm,y=0.8cm]
    \newcommand{\eventDot}[3]{
    \node[circle,draw,fill=black,inner sep=0cm,minimum size=.16cm] (#1) at (#2,#3) {};
    }
    
    \tikzmath{
    \lPi = 9; % length of p_i's line
    \lPjPk = 5; % length of other processes' line
    \yInterP = 2; % inter-process y offset
    \xPredPi = 1; % x coord of p_i's WBP event 
    \xRevPi = 4; % x coord of p_i's revealing round
    \xTPlus1 = \lPi-1; % x coord of round t+1
    \xPiCallout = \xRevPi-3.25;
    \yPiCallout = -2;
    \xPjPkCallout = \xTPlus1+2.7;
    \yPjCallout = -\yInterP+1;
    \xPkCallout = \xTPlus1+3;
    \yPkCallout = -2*\yInterP+1;
    }

    \def\colorCallouts{Periwinkle!30}
    
    % --------- p_i --------- %
    
    \node at (-.5,0) {$p_i$};
    \draw[->] (0,0) -- (\lPi,0);

    \eventDot{predPi}{\xPredPi}{0}
    \eventDot{revPi}{\xRevPi}{0}

    %\node[overlay,cloud callout,aspect=6,cloud puffs=30,fill=yellow!90,callout relative pointer= {(.3cm,.4cm)}] at (1,-1.5) {$\delpredicate(m,w,\view_i)$};
    \node[overlay,rectangle callout,fill=\colorCallouts,callout absolute pointer=(predPi)] at (\xPredPi-1,1) {$\delpredicate(m,w,\view_i)$};

    \node (revPiLabel) at (\xRevPi,1) {$\reveal(w)$};
    \draw[dashed] (revPiLabel) -- (\xRevPi,-.65);

    \node[align=center,overlay,rectangle callout,fill=\colorCallouts,callout absolute pointer=(revPi),callout pointer xshift=.7cm] at (\xPiCallout,\yPiCallout) {$\delpredicate(m'{\neq}m,w'{\geq}w,\view_i)$\\$\implies$ $p_i$ brb-delivers $m$};

    \draw[red] (\xPiCallout-1.5,\yPiCallout+.1) -- (\xPiCallout+1.5,\yPiCallout+.6) (\xPiCallout-1.5,\yPiCallout+.6) -- (\xPiCallout+1.5,\yPiCallout+.1);

    % --------- p_j --------- %
    
    \node at (\lPi-\lPjPk-.5,-\yInterP) {$p_j$};
    \draw[dashed] (\lPi-\lPjPk,-\yInterP) -- (\lPi-\lPjPk+.65,-\yInterP);
    \draw[->] (\lPi-\lPjPk+.65,-\yInterP) -- (\lPi,-\yInterP);

    \eventDot{predPj}{\xTPlus1}{-\yInterP}
    \draw[dotted,thick,->] (predPi) -- (predPj);
    \draw[dotted,thick,->] (revPi) -- (predPj);

    \node[align=center,overlay,rectangle callout,fill=\colorCallouts,callout absolute pointer=(predPj),callout pointer xshift=-.2cm] at (\xPjPkCallout,\yPjCallout) {$\delpredicate(m,w,\view_j)$\\$\delpredicate(m'{\neq}m,w'{\geq}w,\view_j)$};

    \draw[red] (\xPjPkCallout-1.5,\yPjCallout) -- (\xPjPkCallout+1.5,\yPjCallout-.5) (\xPjPkCallout-1.5,\yPjCallout-.5) -- (\xPjPkCallout+1.5,\yPjCallout);

    % --------- p_k --------- %
    
    \node at (\lPi-\lPjPk-.5,-2*\yInterP) {$p_k$};
    \draw[dashed] (\lPi-\lPjPk,-2*\yInterP) -- (\lPi-\lPjPk+.65,-2*\yInterP);
    \draw[->] (\lPi-\lPjPk+.65,-2*\yInterP) -- (\lPi,-2*\yInterP);

    \eventDot{predPk}{\xTPlus1}{-2*\yInterP}
    \draw[dotted,thick,->] (predPi) -- (predPk);
    \draw[dotted,thick,->] (revPi) -- (predPk);

    \node[align=center,overlay,rectangle callout,fill=\colorCallouts,callout absolute pointer=(predPk),callout pointer xshift=-.2cm] at (\xPjPkCallout,\yPkCallout) {$\delpredicate(m,w,\view_j)$\\$\delpredicate(m'{\neq}m,w'{\geq}w,\view_j)$};

    \draw[red] (\xPjPkCallout-1.5,\yPkCallout) -- (\xPjPkCallout+1.5,\yPkCallout-.5) (\xPjPkCallout-1.5,\yPkCallout-.5) -- (\xPjPkCallout+1.5,\yPkCallout);

    % --------- round t+1 --------- %

    \node (tPlus1Label) at (\xTPlus1,1) {round $t+1$};
    \draw[dashed] (tPlus1Label) -- (\xTPlus1,-2*\yInterP-.65);

    % --------- misc --------- %

    \node at (11.5,0) {};
    
\end{tikzpicture}
  
  \caption{If a correct process $p_i$ observes $\delpredicate(m,w,\view_i)$ or reaches round $\reveal(w)$ only with message $m$ as a candidate, then these events are eventually ``propagated'' to all other correct processes (here $p_j$ and $p_k$), at the latest in round $t+1$ in bad cases, and immediately in good cases. \ta{not necessarily, Byzantine processes can send signatures to only some correct processes, delaying the time of the observation of the predicate for some processes}}
  \label{fig:predicate-propagation}
\end{figure}
    
\subsection{\texorpdfstring{\lambdagood}{lambda}-BRB-robustness and BRB guarantees}
\label{sec:lambd-brb-robustn}

The pair $(\delpredicate,\reveal)$ is said to be \emph{\lambdagood-BRB-robust}\ft{TBD: there are probably better options.} if \delpredicate and $\reveal()$ exhibit the following properties when used in Algorithm~\ref{algo:generic:SBRB:pi}, where \lambdagood is an integer that depends on $n$, $t$, and $c$.

The properties are grouped into three blocks.
The first block, \emph{Monotony and Local Conspicuity}, states four intuitive properties on the stability of the \delpredicate predicate and the \reveal function when their weight and view parameters change.
The second block, \emph{Safety}, formalizes the conspicuity and the final visibility of a weight-based predicate, two properties that are essential to prove the safe termination of Algorithm~\ref{algo:generic:SBRB:pi}.
Finally, the third block, \emph{Liveness}, contains a single property that describes the behaviour of \delpredicate in good cases, which determines the good-case latency of Algorithm~\ref{algo:generic:SBRB:pi}.

The tight connection between these properties and BRB broadcast is captured by Theorem~\ref{theo:algo:generic:SBRB:pi:works}, described just afterwards. Theorem~\ref{theo:algo:generic:SBRB:pi:works} states that if a pair $(\delpredicate,\reveal)$ is \lambdagood-BRB-robust (i.e. fulfills the properties listed below), then Algorithm~\ref{algo:generic:SBRB:pi} implements a BRB broadcast that exhibits a  good-case latency of $\maxfn\big(2,\lambdagood\big)$.

In the following  $w,w'\in\mathbb{N}^{+}$ are weights; $m\in\mathcal{M}$ is a message; and $\view$ and $view'$ are sets of valid signature chains.

\begin{itemize}
\item Monotony and Local Conspicuity\ft{I'm not a fan of ``Local Conspicuity'', but could not find anything better}
  \begin{itemize}
  \item \WBPWeightMonoProp: If \delpredicate holds for a given weight, it should also hold for any smaller weight: $\forall w\geq w': \delpredicate(m,w,\view) \Longrightarrow \delpredicate(m,w',\view)$.
  
  \item \WBPViewMonoProp: If \delpredicate holds for a given view, it should also hold for any larger view (in the sense of set inclusion): $\forall view\subseteq view', \delpredicate(m,w,\view) \Longrightarrow \delpredicate(m,w,\view')$.
  
  \item \WBPRevealWeightMonoProp: A heavier predicate should exhibit an earlier revealing round: $\forall w>w': \reveal(w)\leq \reveal(w')$.
  
  \item \WBPLocalConspProp: A message exhibits a weight-$1$ predicate in a view if and only if it appears in the view: $\delpredicate(m,1,\view) \Longleftrightarrow m\in\messageOf(view)$
    % If a correct process $p_i$ observes $\delpredicate(m,w,\view_i)$, then $m$ is known to $p_i$.
  \end{itemize}
\item Safety
  \begin{itemize}
  \item \WBPConspiProp: Consider $p_i$, $p_j$ two correct processes different from $\psender$, and $w$ a weight.
  If $p_i$ observes $\delpredicate(m,w,\view_i)$ during its execution, and $p_j$ executes round $\reveal(w)$, then $m$ is known to $p_j$ at the latest by round $\reveal(w)$.
  
  \item \WBPVisiProp: Consider $p_i$, $p_j$ two correct processes different from $\psender$. If $\psender$ is Byzantine and $p_i$ observes $\delpredicate(m,$ $w,$ $\view_i)$ during its execution, and $p_j$ executes round $t+1$, then $p_j$ observes $\delpredicate(m,$ $w,$ $\view_j)$ at round $t+1$.
  \end{itemize}
\item Liveness
  \begin{itemize}
  \item \WBPLiveProp: If $\psender$ is correct and \brb-broadcasts a message $m$, then all correct processes $p_i\neq\psender$ observe $\delpredicate(m,\wgood,\view_i)$ during round 2, where $\wgood\in\mathbb{N}^{+}$ is a weight such that $\reveal(\wgood)=\lambdagood$.
  \end{itemize}
\end{itemize}

\begin{theorem}\label{theo:algo:generic:SBRB:pi:works}
  If the pair $(\delpredicate, \reveal)$ is \lambdagood-BRB-robust, where \lambdagood is an integer that depends on $n$, $t$, and $c$, then Algorithm~{\em\ref{algo:generic:SBRB:pi}} implements a Synchronous Byzantine Reliable Broadcast object with a worst-case latency of $t+1$ rounds. Furthermore, if the initial sender $\psender$ is correct, correct processes \brb-deliver in at most $\maxfn\big(2,\lambdagood\big)$ rounds.
\end{theorem}

\subsection{Proof of Theorem~\ref{theo:algo:generic:SBRB:pi:works}}
\label{sec:proof-theor-refth}

The proof of Theorem~\ref{theo:algo:generic:SBRB:pi:works} follows from Lemmas~\ref{lemma:generic:Vprop}-\ref{lemma:generic:worst:t+1:rounds}, which follow.

\begin{lemma}
\label{lemma:generic:Vprop}
% If the pair $(\delpredicate, \reveal)$ is (\wgood)-BRB-robust, then
Algorithm~{\em\ref{algo:generic:SBRB:pi}} verifies the 
  {\em \brbVprop} Property.
\end{lemma}%

% \lemmaVprop*

\begin{proof}
  Consider $p_i$ a correct process.
  \begin{itemize}
  \item If $p_i=p_{\orig}$, the \brb-delivery of a message $m$ at line~\ref{line:generic:porig:brdeliver} of Algorithm~\ref{alg:generic:porig:code} trivially implies that $p_i$ has executed Algorithm~\ref{alg:generic:porig:code}, and hence has \brb-broadcast $m$.
    
  \item If $p_i\neq p_{\orig}$, $p_i$ may \brb-deliver a message $m$ either at lines~\ref{line:generic:brdeliver:immediate} or~\ref{line:generic:brdeliver:finalround} of Algorithm~\ref{algo:generic:SBRB:pi}. In both cases, $m$ belongs to some $\knownmessages_{i,R}$ variable computed at line~\ref{line:generic:computing:knownmsg}, and must therefore appear in a signature chain of the form $m\cryptchain p_{i_1} \cryptchain \cdots \cryptchain p_{i_\ell}$ received by $p_i$ at line~\ref{line:generic:view:i:R}. As $p_i$ is correct, it only accepts and processes valid chains of signatures by assumption, in which $m$ is first signed by $p_{\orig}$ (i.e. $p_{i_1}=p_{\orig}$). Since $p_{\orig}$ is correct, and we have assumed signatures to be secure, for $m$ to be signed by $p_{\orig}$, $p_{\orig}$ must have executed line~\ref{line:generic:porig:bcast} of Algorithm~\ref{alg:generic:porig:code}, and must therefore have \brb-broadcast $m$.\qedhere
  \end{itemize}%
\end{proof}

\begin{lemma}%
  %  If the pair $(\delpredicate, \reveal)$ is (\wgood)-BRB-robust, then
  Algorithm~{\em\ref{algo:generic:SBRB:pi}} verifies the 
  {\em \brbNDNprop} Property.%
\end{lemma}

%\lemmaNDNprop*

\begin{proof}
  Trivially, this is because once a correct process executes a \brbdeliver operation (either at line~\ref{line:generic:porig:brdeliver} of Algorithm~\ref{alg:generic:porig:code}, or lines~\ref{line:generic:brdeliver:immediate} or~\ref{line:generic:brdeliver:finalround} of Algorithm~\ref{algo:generic:SBRB:pi}), it terminates its execution, either immediately or at line~\ref{line:generic:stop:ready} in the next round, without invoking \brbdeliver.
\end{proof}

\begin{lemma}
  %  If the pair $(\delpredicate, \reveal)$ is (\wgood)-BRB-robust, then,
  Algorithm~{\em\ref{algo:generic:SBRB:pi}} verifies the 
  {\em \brbLDprop} property.
\end{lemma}

% \lemmaLDprop*

\begin{proof}
  The property trivially follows from the code executed by $\psender$ (Algorithm~\ref{alg:generic:porig:code}). If $p_{\orig}$ is correct it executes Algorithm~\ref{alg:generic:porig:code} to broadcast
    a message $m$, then \brb-delivers its own message at line~\ref{line:generic:porig:brdeliver}.
\end{proof}

\begin{lemma}
  If the pair $(\delpredicate, \reveal)$ is \lambdagood-BRB-robust, then, Algorithm~{\em\ref{algo:generic:SBRB:pi}} verifies the 
  {\em \brbNDYprop} Property.
\end{lemma}

%\lemmaNDYprop*

\begin{proof}~
  \begin{itemize}
  \item If $p_{\orig}$ is correct, $p_{\orig}$ \brb-broadcasts one single message $m$ (Algorithm~\ref{alg:generic:porig:code}), and by \brbVprop (Lemma~\ref{lemma:generic:Vprop}), all correct processes that do \brb-deliver a message only \brb-deliver $m$.
  
  \item If $p_{\orig}$ is Byzantine, consider two correct processes $p_i$ and $p_j$ (both necessarily different from $p_{\orig}$) that each \brb-deliver some message: $p_i$ \brb-delivers $m_i$ and $p_j$ \brb-delivers $m_j$. We distinguish three cases depending on the lines at which $p_i$ and $p_j$ execute \brbdeliver.
    \begin{itemize}
    \item Case 1: Assume $p_i$ and $p_j$ both deliver their respective message at line~\ref{line:generic:brdeliver:immediate} of Algorithm~\ref{algo:generic:SBRB:pi}. Due to the condition at line~\ref{line:generic:test:knowmessage}, there exist two rounds $R_i$ and $R_j$ such that the following holds
      \begin{align*}
        \knownmessages_{i,R_i} &= \{m_i\} \land \Exists w_i \in \mathbb{N}^{+}: \big(\delpredicate(m_i,w_i,\view_i) \land R_i \geq \reveal(w_i)\big),
        \intertext{and}
        \knownmessages_{j,R_j} &= \{m_j\} \land \Exists w_j \in \mathbb{N}^{+}:  \big(\delpredicate(m_j,w_j,\view_j) \land R_j \geq \reveal(w_j)\big).
      %\end{array}
      \end{align*}
      Without loss of generality, assume $w_i\geq w_j$. By \WBPRevealWeightMonoProp, $\reveal(w_i)\leq \reveal(w_j)$, which leads by case assumption to $R_j$ $\geq$ $ \reveal(w_i)$.
      Process $p_j$ therefore executes round $\reveal(w_i)$, and \WBPConspiProp applies to $\delpredicate(m_i,$ $w_i,$ $view_i)$ that $p_i$ observes.
      We conclude that $m_i$ is known to $p_j$ by round $\reveal(w_i)$, i.e. formally $m_i\in\knownmessages_{j,\reveal(w_i)}$.
      %\ft{I've put back the subscript $_{i,R}$ to $\knownmessages$ in Alg.~\ref{algo:generic:SBRB:pi}, as this seemed more direct than explaining the notation here. To be discussed.}
      Since $\view_i$ keeps growing with each passed round, $R_j\geq \reveal(w_i)$ implies $\knownmessages_{j,\reveal(w_i)}\subseteq \knownmessages_{j,R_j}$, and therefore $m_i\in\knownmessages_{j,R_j}$.
      Since $\knownmessages_{j,R_j}=\{m_j\}$ by case assumption, this leads to $m_i=m_j$, proving the case.

    \item Case 2: Assume $p_i$ and $p_j$ both \brb-deliver their respective message at line~\ref{line:generic:brdeliver:finalround} or~\ref{line:generic:brbdeliver:bot} of Algorithm~\ref{algo:generic:SBRB:pi}, during round $t+1$.
    Let us consider the two following sets, defined at round $t+1$:
    \begin{align*}
        \weights_i &= \{w \in \mathbb{N}^{+} \mid \Exists m \in \knownmessages_{i,t+1}: \delpredicate(m,w,\view_i)\},
      \intertext{and}
      \weights_j &= \{w \in \mathbb{N}^{+} \mid \Exists m \in \knownmessages_{j,t+1}: \delpredicate(m,w,\view_j)\}.
    \end{align*}
      Consider $w\in \weights_i$, and $m \in \knownmessages_{i,t+1}$ a message such that $\delpredicate(m,w,\view_i)$.
      Because $\psender$ is Byzantine (by assumption), \WBPVisiProp applies and $\delpredicate(m,w,\view_i)$ for $p_i$ implies $\delpredicate(m,w,\view_j)$ for $p_j$ at round $t+1$. By \WBPWeightMonoProp and \WBPLocalConspProp, $\delpredicate(m,w,\view_j)$ at round $t+1$ implies $\delpredicate(m,1,\view_j)$ (since $w\geq 1$ by construction), $m\in \knownmessages_{j,t+1}$, and therefore that $w\in \weights_j$.
      Inverting $p_i$ and $p_j$ leads to $\weights_i=\weights_j$, and therefore to $\maxfn(\weights_i)=\maxfn(\weights_j)$.
      Using this last equality, we conclude that either $p_i$ and $p_j$ both have empty $\weights_i$ and $\weights_j$ sets and get to the else branch at line~\ref{line:generic:brbdeliver:bot} and both brb-deliver $\bot$, or they both pass the condition at line~\ref{line:generic:cond:weight:not:empty}, in which case they must both obtain the same $\candidatemsg_i$ and $\candidatemsg_j$ sets, and therefore \brb-deliver the same message at line~\ref{line:generic:brdeliver:finalround}.
    \item Case 3:
Assume $p_i$ \brb-delivers $m_i$ at line~\ref{line:generic:brdeliver:immediate} of Algorithm~\ref{algo:generic:SBRB:pi} during some round $R_i$, and $p_j$ \brb-delivers $m_j$ at line~\ref{line:generic:brdeliver:finalround} or~\ref{line:generic:brbdeliver:bot} of the same algorithm during round $t+1$.
Due to the condition at line~\ref{line:generic:test:knowmessage}, there exists a weight $w_i\in\mathbb{N}^{+}$ such that the following holds
      \begin{gather}
        \knownmessages_{i,R_i} =\{m_i\} \land \delpredicate(m_i,w_i,view_i) \land R_i \geq \reveal(w_i). \label{eq:generic:knownmsg:by:pi}
      \end{gather}
      As in Case 2, let us consider the following set defined at round $t+1$ at $p_j$:
      $$\weights_j = \{w \in \mathbb{N}^{+} \mid \Exists m \in \knownmessages_{j,t+1}: \delpredicate(m,w,\view_j)\}.$$
      % Because of Lemma~\ref{lemma:GCL:predicate:implies:m:known}, $\knownmessages_{i,R_i} =\{m_i\}$ implies that $m_i\in\knownmessages_{j,R_i}$. As $R_i\leq t+1$, the definitions of $\knownmessages_{j,-}$ and $\view_j$ yields that $\knownmessages_{j,R_i}\subseteq \knownmessages_{j,t+1}$, and therefore that $m_i\in \knownmessages_{j,t+1}$. Furthermore 

      As in Case 2, \WBPVisiProp applies and $\delpredicate(m_i,w_i,\view_i)=\ttrue$ at $p_i$ implies $\delpredicate(m_i,w_i,\view_i)=\ttrue$ at $p_j$ at round $\finalround$.
      This fact and \WBPLocalConspProp further implies $m_i \in \knownmessages_{j,t+1}$, and therefore that $w_i\in \weights_j$ (and as $\weights_j$ is not empty, $p_j$ cannot brb-deliver at line~\ref{line:generic:brbdeliver:bot}).
      This last inclusion yields that $\maxfn(\weights_j)\geq w_i$ at line~\ref{line:generic:wmaxj:computed} of Algorithm~\ref{algo:generic:SBRB:pi}.

      Because $m_j$ is \brb-delivered by $p_j$ at line~\ref{line:generic:brdeliver:finalround}, we have by construction $\delpredicate(m_j,$ $\maxfn(\weights_j),$ $\view_j)\!=\ttrue$ at line~\ref{line:generic:wmaxj:computed} of $p_j$.
      By \WBPWeightMonoProp, $\delpredicate(m_j,$ $\maxfn(\weights_j),$ $\view_j)$ and $\maxfn(\weights_j)\geq w_i$ %% and  $\delpredicate(m_j,$ $\maxfn(\weights_j),$ $\view_j)\!=\ttrue$
      imply $\delpredicate(m_j,$ $w_i,$ $\view_j)\!=\ttrue$ at $p_j$. Applying \WBPConspiProp, and the fact that $R_i\geq\reveal(w_i)$, this last statement implies that 
 $m_j\in\knownmessages_{i,R_i}$ at round $R_i$ at $p_i$. Combined with (\ref{eq:generic:knownmsg:by:pi}), this leads to $m_j=m_i$, proving the case and concluding the lemma. \qedhere
    \end{itemize}%
  \end{itemize}%
%\vspace{-2\baselineskip}~
  
\end{proof}

\begin{lemma}
\label{lemma:generic:porig:correct:max:nb:round}
  If the pair $(\delpredicate, \reveal)$ is \lambdagood-BRB-robust, where \lambdagood is an integer that depends on $n$, $t$, and $c$, and if $\psender$ is correct, then correct processes brb-deliver the message $m$ brb-broadcast by $\psender$ in at most $\maxfn\big(2,\lambdagood\big)$ rounds.
\end{lemma}

%\lemmaGoodCaseLatency*

\begin{proof}
  If $\psender$ is correct and brb-broadcasts a message $m$, it \brb-delivers its own message in round 1 (Algorithm~\ref{alg:generic:porig:code}). By \WBPLiveProp  all other correct processes $p_i\neq\psender$ observe $\delpredicate(m,\wgood,\view_i)$ during round 2. By \WBPViewMonoProp, and given that $\view_i$ can only grow in Algorithm~\ref{algo:generic:SBRB:pi}, $\delpredicate(m,\wgood,\view_i)$ remains true during all subsequent rounds $R\geq 2$ that $p_i$ executes.
  In addition, as $\psender$ is correct and signatures are secure, $\knownmessages_{i,R}$ does not contain any other message than $m$.
  As a result, at the latest in round $\maxfn(2,\reveal(\wgood))=\maxfn(2,\lambdagood)$, the condition of line~\ref{line:generic:test:knowmessage} becomes true for the message $m$ and the weight $\wgood$, and $p_i$ delivers $m$.
\end{proof}

\begin{lemma}\label{lemma:generic:worst:t+1:rounds}
  In the worst case, a correct process executing Algorithm~{\em\ref{algo:generic:SBRB:pi}} \brb-delivers a message in $t+1$ rounds.
\end{lemma}

\begin{proof}
By construction of Algorithm~\ref{algo:generic:SBRB:pi}, processes are guaranteed to \brb-deliver at the latest in round $t+1$ (through the condition at line~\ref{line:generic:reached:finalround}), either at line~\ref{line:generic:brdeliver:finalround} or~\ref{line:generic:brbdeliver:bot} (if they have not brb-delivered early at line~\ref{line:generic:brdeliver:immediate}).
% This worst case occurs when the \psender is Byzantine, and Byzantine processes remain silent during the first $t-1$ rounds, and disclose in round $t$ to a single correct process $p_i$ a chain of length $t$ containing all Byzantine processes, and backing a message $m$.
% The remaining correct processes (assuming $c\geq 2$) only hear of $m$ in round $t+1$, and deliver it in this round.
\end{proof}

\begin{lemma}
\label{lemma:generic:GDprop}
  Algorithm~{\em\ref{algo:generic:SBRB:pi}} verifies the 
  {\em \brbGDprop} property.
\end{lemma}

%\lemmaGDprop*

\begin{proof}
This property derives trivially from Lemma~\ref{lemma:generic:worst:t+1:rounds}: all correct processes brb-deliver a message at the latest in round $t+1$. \qedhere
  % \begin{itemize}
  % \item If $\psender$ is correct, then using Lemma~\ref{lemma:generic:porig:correct:max:nb:round} all correct processes execute \brbdeliver, the lemma is verified. 
  % \item If $\psender$ is Byzantine, consider $p_i$ and $p_j$ two correct processes ($\{p_i,p_j\}\cap \{\psender\}=\emptyset$) and assume $p_i$ \brb-delivers some message $m$.

  %   If $p_j$ does not execute round $t+1$, by construction $p_j$ must has stopped at line~\ref{line:generic:stop:ready}, which implies that $p_j$ executed line~\ref{line:generic:brdeliver:immediate} and \brb-delivered some message, proving the lemma.

  %   If $p_j$ does execute round $t+1$, whether $p_i$ \brb-delivers $m$ at line~\ref{line:generic:brdeliver:finalround} or at line~\ref{line:generic:brdeliver:immediate}, the \brb-delivery implies that $p_i$ observes $\delpredicate(m,w,\view_i)=\ttrue$ for some $w\in \mathbb{N}^{+}$ at some point of its execution. Because $\psender$ is Byzantine \WBPVisiProp applies, and $p_j$ therefore observes $\delpredicate(m,$ $w,$ $\view_j)$ at round $t+1$. By \WBPLocalConspProp and \WBPWeightMonoProp, this leads to 
  %   $\delpredicate(m,$ $1,$ $\view_j)$ (since $w\geq 1$) and $m\in\messageOf(view_j)$. Therefore at $p_j$ during round $t+1$, $\weight_{j}\neq\emptyset$, which implies $\candidatemsg_j\neq\emptyset$, and means $p_j$ \brb-delivers some message at line~\ref{line:generic:brdeliver:finalround} during round $t+1$, concluding the proof.
  %   \qedhere
  % \end{itemize}
\end{proof}

\subsection{Revisiting Lamport, Shostak, and Pease's algorithm (BRB-LSP)}
\label{sec:revis-lamp-shost}

BRB-LSP~\cite{LSP82} can be expressed in the framework of Algorithm~\ref{algo:generic:SBRB:pi} by choosing the following definitions and values for \delpredicate, \reveal, and \lambdagood:
\newcommand{\LSP}{\mathsf{LSP}}%
\begin{align*}%
  \delpredicate_{\LSP}(m,w,\view) &\isdefinedas \big(m\in\messageOf(\view)\big),\\
  \reveal_{\LSP}(w) &\isdefinedas t+1,\\
  \lambdagood^{\LSP} &\isdefinedas t+1.
\end{align*}%

With the above definition of $\delpredicate$ and $\reveal$, Algorithm~\ref{algo:generic:SBRB:pi} systematically \brb-delivers in round $t+1$ (since the condition at line~\ref{line:generic:test:knowmessage} can only become true in round $t+1$). Furthermore, when $p_i$ only knows one message $m$ in round $t+1$, the first branch of the ``delivery'' if block at line~\ref{line:generic:brdeliver:immediate} is equivalent to its else branch at lines \ref{line:generic:Wi:set}-\ref{line:generic:brbdeliver:bot}.
In all cases, $p_i$ therefore collects in round $t+1$ all the messages it has received, and chooses one of them in a deterministic manner, thus reproducing BRB-LSP.

$\delpredicate_{\LSP}$ and $\reveal_{\LSP}$ correspond to a border example of the use of our generic algorithm, since the weight $w$ plays no part in their definition. However, the pair $(\delpredicate_{\LSP},\reveal_{\LSP})$ does fulfill the prerequisites of $(t+1)$-BRB-robustness required to apply Theorem~\ref{theo:algo:generic:SBRB:pi:works}, as the following theorem shows.

\begin{theorem}\label{the:LSP:pair:t+1:robust}
  The pair $(\delpredicate_{\LSP},\reveal_{\LSP})$ is $(t+1)$-BRB-robust.
\end{theorem}

\begin{proof}
  \WBPWeightMonoProp and \WBPRevealWeightMonoProp trivially follow from the fact that neither $\delpredicate_{\LSP}$ nor $\reveal_{\LSP}$ depend on $w$.
  Similarly, \WBPViewMonoProp and \WBPLocalConspProp directly result from the definition of $\delpredicate_{\LSP}$.

  \WBPConspiProp and \WBPVisiProp hinge on the central intuition underpinning BRB-LSP. First note that $\delpredicate_{\LSP}(m,w,\view_i)=\ttrue$ for some process $p_i$, message $m$, and weight $w$ is equivalent to stating that $p_i$ knows message $m$, or equivalently that $p_i$ has received some valid chain containing $m$.
  
  Consider a correct process $p_i\neq\psender$ that has received some chain containing $m$. If $p_i$ received $m$ for the first time before round $t+1$, then by construction of Algorithm~\ref{algo:generic:SBRB:pi}, it has sent it to all other correct processes, which must therefore also know $m$. If $p_i$ received $m$ for the first time during round $t+1$, then the chain carrying $m$ must contain $t+1$ distinct processes (to be valid), and must therefore contain at least one correct process $p_k$ that must have sent $m$ to all other correct processes. These observations yield the \WBPConspiProp and \WBPVisiProp properties.

  Finally, \WBPLiveProp is trivially fulfilled by definition of $\delpredicate_{\LSP}$.
\end{proof}

%
%\section{A deterministic synchronous BRB algorithm}
%\label{sec:determ-synchr-brb}

%%%%%%%%%%%%%%%%%%%%%%%%%%%%%%%%%%%%%%%%%%%%%%%%%%%%%%%%%%%%%
\section{A Deterministic Good-Case BRB Latency in \texorpdfstring{$\maxfn\big(2,t+3-w\big)$}{max(2,t+3-w)} Rounds}
\label{sec:descr-algor}

\subsection{Overview}
\label{algo-overview}

Theorem~\ref{theo:algo:generic:SBRB:pi:works} states that a deterministic synchronous BRB algorithm that exhibits a good case latency of $\maxfn\big(2,t+3-w\big)$ can be obtained simply by finding a weight-base predicate and revealing round function that are $\maxfn\big(2,t+3-w\big)$-BRB-robust. Constructing such a pair is however not immediately obvious. We present such a predicate in this section (called \emph{weight-based predicate with good-case latency}, or \GCLpredicate for short) by exploiting a pattern revolving around what we have called a ``revealing chain''.

\paragraph{Weights and revealing chains} The weight-based predicate we propose counts the number of processes whose signature appears within the first two positions of the valid chains a process $p_i$ has received. These processes are said to be \textit{backing} $m$ in $p_i$'s view, and their number is the predicate's weight, $w$.
%% revolves around the notion of \emph{certificate}, which can be informally described as a set of signature chains for a given message $m$ that fits  a particular pattern. The \emph{weight} of a certificate is defined as the number of processes whose signature appears within the first two positions of some chains of the certificate. These processes are said to be \textit{backing} $m$ in the certificate.

% FT09Jan23: copied from earlier version
%% We constrain how a certificate might be propagated to limit how long Byzantine processes can hide a valid certificate from correct processes.
%% A given certificate for a message $m$ has a ``weight'' representing how many processes are ``backing'' $m$.
%% To back a message $m$, a process must have witnessed it at the latest by the end of round 2.

Just counting and propagating the round-$1$ or -$2$ signatures that correct processes observe is, however, not enough, as it does not prevent Byzantine processes from selectively revealing some round-2 signatures at the very last moment (round $t+1$ in our case), thus preventing correct processes from \brb-delivering earlier using this information only.
The predicate we use therefore adds an additional constraint that limits the disruption power of Byzantine processes, and provides the conspicuity property required by Algorithm~\ref{algo:generic:SBRB:pi}: a predicate of weight $w$ must contain a ``\emph{revealing chain}'' $m\cryptchain \cryptchain \gamma$  whose makeup must ``differ sufficiently'' from the backing processes documented by the predicate. ``Differ sufficiently'' means that the processes from position~$3$ until position $t+3-w$ of this revealing chain (shown in red in Figure~\ref{fig:GCL_predicate:pattern}) should not be backing processes.
%% such that processes whose signature appears between positions $3$ and $t+3-w$ of this revealing chain are not backing $m$ in the certificate.
%% More formally, these processes cannot appear in position $2$ in any of the certificate's chains.
%% If the revealing chain contains fewer than $t+3-w$ processes, then only those processes that have signed the chain beyond round $3$ need to be considered.
%% (Formally, this condition is encoded in a function $\GCLpredicate_i(\cdot)$, whose pseudo-code is presented in the next section.)

This constraint limits what Byzantine processes can do when the sender is Byzantine and ensure this predicate is both conspicuous and finally visible, %% , two key properties required to make it \lambdagood-BRB-robust (Section~\ref{sec:lambd-brb-robustn}),
and thus usable within Algorithm~\ref{algo:generic:SBRB:pi}.
The revealing chain does not prevent Byzantine processes from colluding to forge competing predicates for different messages in bad cases (i.e. when $\psender$ is Byzantine). However, Byzantine processes can only use up to $t$ signatures and must decide whether to invest these $t$ signatures in the backing part of each predicate (thus increasing the predicate's weight) or in their revealing chain (thus delaying the time at which the message of a forged predicate must be revealed to correct processes, but reducing the predicate's weight).

\paragraph{Predicate conspicuity} %The constraint that the $t+2-w$ truncation of $\sequencei$ does not appear in $E_{i,R}$ (formally $\set(\subchain_{t+2-w}(\sequencei))\cap E_{i,R}=\emptyset$) is instrumental to force malicious processes to reveal the potential existence of a predicate early enough. More precisely, we shall see that if a correct process observes a predicate of weight $w$ for a message $m$, then necessarily all correct processes that have not stopped earlier will have received at least one chain containing $m$ by round $R_w=t+3-w$. 
The position $t+3-w$ of the revealing chain enforces the conspicuity of the predicate and yields the good-case latency $\maxfn(2,t+3-c)$. This is because the signatures from positions $3$ to $t+3-w$ correspond to $(t+3-w)-3+1=t+1-w$ processes. Added to the $w$ processes backing the predicate ($W_{i,R}$ in Figure~\ref{fig:GCL_predicate:pattern}), this represents $t+1-w+w=t+1$ processes. These $t+1$ processes must contain a correct process; therefore, Byzantine processes that seek to forge a predicate for a message $m$ must include the signature of a correct process at the latest in round $t+3-w$. This ensures $m$ become visible to all other correct processes by round $t+3-w$, the \emph{revealing round} of the predicate in the terms of Section~\ref{sec:generic-weight-based}. %% We call the round $R_w=t+3-w$ the \emph{conspicuity round} for weight $w$, and this property \emph{Predicate Conspicuity}. % Similarly, round $R_w=t+3-w$ is the \emph{conspicuity round} of the predicates of weight $w$.

%% By contrapositive, predicate conspicuity allows correct processes to ascertain the \emph{nonexistence} of a predicate of a given weight for a message. This ability to be sure that a given predicate does not exist, and the ability to propagate predicates that do, are the key ingredients that allow our algorithm to terminate (much) faster than other chain-based deterministic algorithms~\cite{DS83,FN09,LSP82} in good cases, more precisely in $\maxfn(2,t+3-c)$ rounds, where $c=n-f$ is the number of effective correct processes.

%% This constraint limits what Byzantine processes can do when the sender is Byzantine, and allows correct processes to rely on an early delivery condition that is safe both in good and bad cases. When $p_sender$ is Byzantine (bad case), Byzantine processes could try to disseminate a first message $m$, while hiding for as long as possible the existence of a second message $m'$, together with an appropriate predicate for $m'$ that is at least as strong as that of $m$ (thus making it unsafe to br-deliver $m$ without knowing $m'$). However, to do this Byzantine processes only have access to $t$ signatures. They must decide whether to invest these $t$ signatures in the backing part of $ m'$'s predicate (thus increasing $ m'$'s weight) or in the revealing chain of the predicate (thus delaying the time at which $m'$ will need to be revealed to correct processes, but reducing $ m'$'s weight).

\begin{figure}[tb]
  \centering
\begin{tikzpicture}[x=1cm,y=0.8cm]
    \newcommand{\sigDot}[3]{
    \node[circle,draw,fill=black,inner sep=0cm,minimum size=.16cm] (#1) at (#2,#3) {};
    }
    
    \tikzmath{
    % WEIGHT SET
    \xR2 = .8; % x coord
    \wWS = 1.3; % width
    \hWS = 2; % height
    \interWS = .45; % inter-process offset
    % CHAINS
    \interSig = .5;
    % TRUNCATION
    \wTrunc = 1.4;
    \hTrunc = .4;
    \xTrunc = \xR2 + \wTrunc/2 + .4;
    % LABELS
    \xLLabels = -2.1;
    \xRLabels = 6.5;
    }
    
    \def\colorSender{black}
    \def\colorRII{Green}
    \def\colorTrunc{red}
    
    % --------- sender --------- %
    
    %\node[\colorSender] (snd) at (0,0) {$m{:}p_\orig$};
    \sigDot{snd}{0}{0}

    % --------- chains --------- %

    % index / y coord / nb sigs
    \foreach \nChain/\yChain/\nSigs in {1/2*\interWS/5, 2/\interWS/6, 3/0/4, 4/-\interWS/3, 5/-2*\interWS/5}
    {
        \sigDot{rcv\nChain}{\xR2}{\yChain}
        \foreach \iSig in {1,...,\nSigs}
            \sigDot{chn\nChain sig\iSig}{.1+\xR2+\iSig*\interSig}{\yChain};
        \draw[thick] (snd) -- (\xR2,\yChain) -- (chn\nChain sig\nSigs);
    }

    % --------- weight set --------- %
    
    \node[draw=\colorRII,rounded corners,very thick,minimum width=\wWS cm,minimum height=\hWS cm] (WS) at (\xR2/2,0) {};
    
    % --------- subchain --------- %
    
    \node[draw=\colorTrunc,very thick,minimum width=\wTrunc cm,minimum height=\hTrunc cm] (trunc) at (\xTrunc,-2*\interWS) {};

    % --------- revealing round --------- %
    
    \draw[dashed] (\xTrunc+\wTrunc/2+.1,3.2*\interWS) -- (\xTrunc+\wTrunc/2+.1,-3.2*\interWS);

    \node[align=center] (wsLabel) at (\xTrunc+\wTrunc/2+.15,4*\interWS) {Revealing round};
    
    % --------- left labels --------- %
    
    \node[align=center] (sndLabel) at (\xLLabels,.8) {Sender's signature\\($m \cryptchain p_\orig$)};
    \draw[dashed,->] (sndLabel.east) to[bend left] (snd);
    
    \node[\colorRII,align=center] (wsLabel) at (\xLLabels,-.8) {Set of weight\\signatures $W$};
    \draw[dashed,->,\colorRII] (wsLabel) to[bend right=10] (WS);

    % --------- right labels --------- %

    \node (revChainLabel) at (\xRLabels,0) {Revealing chain $\gamma$};
    \draw[dashed,->] (revChainLabel.west) to[in=70,out=170] ([xshift=-.25cm]chn5sig5);
    
    \node[\colorTrunc] (subchainLabel) at (\xRLabels,-1) {$\subchain(\gamma,3,t+3-|W|)$};
    \draw[dashed,->,\colorTrunc] (subchainLabel.south west) to[bend left=50] (trunc.south);
\end{tikzpicture}
  
  \caption{The pattern of signature chains forming a \GCLpredicate predicate of weight $|W|=6$ at round $R$ for message $m$ at $p_i$ in a setting with $t=8$. The predicate must verify $\subchain(\gamma,3,t+3-|W|) \cap W=\emptyset$, which ensures its conspicuity (Lemma~\ref{lemma:GCL:WBPConspiProp}).}
  \label{fig:GCL_predicate:pattern}
\end{figure}
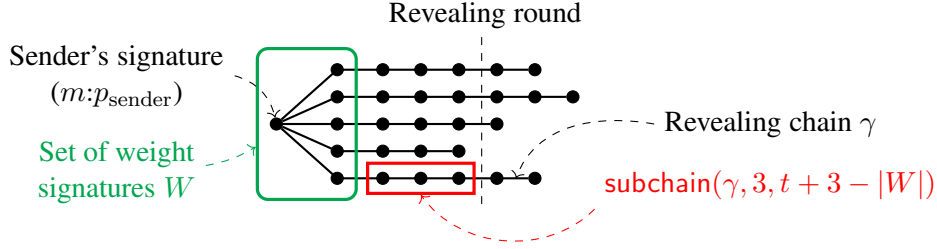

\paragraph{An example of \GCLpredicate} Figure~\ref{fig:GCL_predicate:pattern} shows a $\GCLpredicate$ predicate of weight $w=6$ for a message $m$ observed by $p_i$ at round $R$ in a setting with $t=8$: each horizontal line represents a chain of signatures that starts with $m\cryptchain \psender$, the green rectangle ($W$) represents processes that have signed $m$ in round 1 or 2 (and are therefore backing $m$), and $m\cryptchain \gamma$ is the revealing chain, such that the process appearing from position $3$ to $t+3-w$ ($= t-3$ here) in $m\cryptchain \gamma$ do not appear in $W$.
The corresponding revealing round is defined as $\revealGCL(w) = \maxfn(2,t+3-w)$.
Because $w=6$ and $t=8$, the revealing round is 5 in this example. % as we have $ t=8 \iff \maxfn(2,t+3-w)=5$.
This function is discussed in more detail in Section~\ref{sec:wbp-gcl-predicate}.

It is important to distinguish the round upon which the \GCLpredicate is observed, from the revealing round and the round where the message is brb-delivered.
These can be three entirely different rounds.
In this example, the \GCLpredicate is observed at round 8 when $p_i$ learns about the second chain from the top, but the revealing round is $\revealGCL(6)=t+3-6=5$.
However, it is possible to have the inverse scenario, where the \GCLpredicate is observed before the revealing round, in which case the process must wait in order to brb-deliver the message.
And in the case where multiple distinct messages are observed before the revealing round is attained, the process must wait for round $t+1$ to brb-deliver one of these messages.

%% The predicate depicted in Figure~\ref{fig:GCL_predicate:pattern} is of weight $w=6$, as it proves that $6$ distinct processes are backing $m$ (shown in the green rectangle), i.e. they have signed a chain containing $m$ in round $1$ (for $\psender$) or $2$ (for the others).

\paragraph{A special case: delivery in round 2}
A special case occurs when the weight of a predicate reaches $w=t+1$.
When this happens, any process $p_i$ observing the predicate knows that one of the processes of $W$ is correct and, therefore, that all correct processes must have received a chain containing $m$ by the end of round 2 (the revealing round for the weight $t+1$).
Conversely, if $p_i$ has not received any chain containing a message $m'$ by round 2, $p_i$ knows that a predicate of weight $t+1$ cannot possibly exist for $m'$.
As a result, a correct process that observes a predicate a weight $t+1$ for $m$ and is not aware of any other message $m'\neq m$ by round 2 can safely \brb-deliver $m$, as no other message will be able to ``beat'' $m$ with a heavier predicate, even if the sender $\psender$ is Byzantine.

\paragraph{Weak non-intersecting quorums}
The reasoning for $w=t+1$ mirrors the mechanism of intersecting quorums used in asynchronous systems and requires a majority of correct processes (or $n>2t$) to be guaranteed to occur when the sender is correct.
The proposed predicate mechanism leverages the additional guarantees that a synchronous system brings to generalize this idea to weaker non-intersecting ``quorums'', whose ability to trigger a \brb-delivery decision requires additional temporal information (waiting until the revealing round $t+3-w$).

\subsection{The \texorpdfstring{\GCLpredicate}{WBP-GCL} Predicate}
\label{sec:wbp-gcl-predicate}

\begin{algorithm}[tb]
%  \small
\DontPrintSemicolon
  \IsDefinedAs{$\certificate(m,w,\view)$}{
    $\Exists m\cryptchain\gamma\in\view \text{ such that, when noting }$\label{line:m:gamma:in:view}\Comment*{$m\cryptchain\gamma$ is the revealing chain.}

    \hspace{2em}$S\;\isdefinedas\big\{q\in\Pi\mid q\text{ backs $m$ in round $1$ or $2$ in $\view$}\big\},$ and\label{line:certif:set:S}

    \hspace{2em}$W\isdefinedas S\setminus \set\big(\subchain(\gamma,3,t+3-w)\big),$
    the following holds
    
    \hspace{2em}$|W|\geq w$. \label{line:certif:threshold}\Comment*{At least $w$ processes back $m$ in the first 2 rounds in \view.}
  }
  \smallskip
  
  \lIsDefinedAs{$\revealGCL(w)$}{$\maxfn(2,t+3-w)$.}
  
  \caption{The weight-based predicate $\certificate$ and its associated revealing function $\revealGCL$. The pair $(\certificate,\revealGCL)$ ensures delivery in $\maxfn\big(2,t+3-w\big)$ rounds in good cases when used in Algorithm~\ref{algo:generic:SBRB:pi}.}
  \label{algo:GCL:predicate}
\end{algorithm}

This paper's main contribution, a synchronous deterministic Byzantine Reliable Broadcast algorithm with a good case latency of $\maxfn\big(2,t+3-w\big)$, is obtained by injecting into the generic code of Algorithm~\ref{algo:generic:SBRB:pi} the predicate and reveal function $(\GCLpredicate,\revealGCL)$ shown in Algorithm~\ref{algo:GCL:predicate}.

%% FT03May22: earlier version, protecting crash processes
% If this is the case, $m$ is marked for delivery, and the flag $\ready_i$ toggled to trigger \brb-delivery in the next round at line~\ref{line:brdeliver:ready}. (This delay of one round is required to ensure crashed processes benefit from the $\brNDYprop$ and $\brGDprop$ properties, see footnote~\ref{fn:comment:on:crashed:proc}.)

$\GCLpredicate(m,w,\view)$ first considers $S$, the set of all processes that have signed\footnote{For simplicity, we say that a process $p$ signed $m$ during round $r$ in \view to mean that \view contains a chain backing $m$ in which $p$'s signature appears in $r$\textsuperscript{th} position.
Formally, both formulations are equivalent for correct processes, but they are not for Byzantine processes, as they can sign chains whenever they wish or even use the signature of other Byzantine processes.} the message $m$ either during round $1$ or $2$. 
$\GCLpredicate(m,w,\view)$ is true if $\view$ contains a ``revealing chain'' backing $m$, (noted $m\cryptchain \gamma$) such that after removing from $S$ all processes appearing between the position $3$ and $t+3-w$ of $\gamma$, the resulting set $W$ contains at least $w$ distinct processes that have signed $m$ in round $1$ or $2$ (see Figure~\ref{fig:GCL_predicate:pattern} and 
Section~\ref{algo-overview}.)

The existence of $\gamma$ (the revealing chain) is essential to provide the \WBPConspiProp property of $(\GCLpredicate,\revealGCL)$. Intuitively, the reason lies in the choice of the boundaries of the subchain $\subchain(\gamma,3,t+3-w)$.
\begin{itemize}
\item If $\gamma$ contains strictly less than $t+3-w$ processes and appears in the view $\view_i$ of a correct process $p_i$, then $p_i$ must have known $m$ before round $t+3-w$, and will have informed all other correct processes of $m$'s existence at the latest by round $t+3-w$.
\item If, on the other hand, $\gamma$ contains $t+3-w$ processes or more, then $\subchain(\gamma,3,t+3-w)$ will contain $(t+3-w)-3+1=t+1-w$ distinct processes. In that case, line~\ref{line:certif:threshold} ensures that $W\cup \subchain(\gamma,3,t+3-w)$ contains more than $w+(t+1-w)=t+1$ processes, i.e. does contain at least one correct process. By construction, this correct process will have signed and propagated a chain backing $m$ at the latest by round $t+3-w$.
\end{itemize}

\section{Proof of Correctness}

This Section proves the following Theorem and Corollary.

\begin{theorem}\label{theo:GCL:t+3-c:BRB:robust}
  The pair $(\GCLpredicate,\revealGCL)$ is $\maxfn(2,t+3-c)$-BRB-robust, where $c$ is the effective number of correct processes.
\end{theorem}

\begin{corollary}\label{coro:GCL:BRB:goodcase}
  The use of $(\GCLpredicate,\revealGCL)$ in Algorithm~\ref{algo:generic:SBRB:pi} implements a synchronous deterministic Byzantine Reliable Broadcast algorithm with a worst-case latency of $t+1$ and a good case latency of $\maxfn\big(2,t+3-w\big)$.
\end{corollary}

%% The combination of Theorem~\ref{theo:GCL:t+3-c:BRB:robust} with Theorem~\ref{theo:algo:generic:SBRB:pi:works} proves that 

\begin{description}
\item[Remark]
Note that if $n>2t$, then because $c\geq n-t$, we have $c\geq t+1$, and $\maxfn(2,t+3-c)=2$, all correct processes deliver in at most $2$ rounds when the sender is correct.
\end{description}

%% \todo{With the new Section~\ref{sec:fram-constr-determ}, the next theorem is no longer needed in this formulation: instead we need to show that \GCLpredicate and $t+3-w$ are $c$-BRB-robust, in order to use Theorem~\ref{theo:algo:generic:SBRB:pi:works}.}
%% \begin{theorem}\label{theo:algo:SBRB:works}
%%   Algorithm~{\em\ref{algo:SBRB:pi}} implements a Synchronous Byzantine Reliable Broadcast object. If the initial sender $p_{\orig}$ is correct, correct processes \brb-deliver in at most $\maxfn(2,t+3-c)$ rounds, where $c$ is the effective number of correct processes.
%% \end{theorem}

%% \begin{lemma}
%%   If $p_i$ is a correct process, and $\GCLpredicate_j(m,t+3-R)=\ttrue$ at some round $R\in[2..\finalround]$ then
%%   \begin{align}
%%     &\exists \rmwi'\in[2..\finalround], \exists \sequencei' \in (\Pi\setminus \{p_i\})^{\rmwi'-1}:\nonumber\\
%%     &\hspace{4em}m\cryptchain p_{\orig}\cryptchain \sequencei' \in \view_i[\rmwi']\: \wedge\label{line:chain:no:pi}\\
%%     &\hspace{4em}\textstyle\left|\left\{q\in\Pi\setminus \subchain_{t+2-w}(\sequencei')\,\left|\, m\cryptchain p_{\orig}\cryptchain q \in \subchain_2\left(\bigcup_{r'\in[2..\finalround]} \view_i[r']\right)\right.\right\}\right|\geq w-2.\label{line:certif:threshold:chain:no:pi}
%%       \end{align}
%% \end{lemma}

\subsection{Preliminary lemmas}

\begin{lemma}\label{lemma:GCL:WBPWeightMonoProp}
  The predicate \GCLpredicate fulfills the \WBPWeightMonoProp property defined in Section~{\em\ref{sec:lambd-brb-robustn}}.
\end{lemma}

\begin{proofoverview}
  The lemma follows from the fact that the weight $w$ of a predicate $\GCLpredicate(m,w,\view)$ counts processes that back $m$ in round $1$ or $2$. A (small) technical difficulty is caused by the revealing chain ($m\cryptchain\gamma$ in Algorithm~\ref{algo:GCL:predicate}), which grows when considering a smaller weight $w'\leq w$.
  This growth is, however, bounded by the weight difference $w-w'$, which yields the lemma.
\end{proofoverview}

\begin{detailedproof}
Assume $\GCLpredicate(m,w,\view)$ holds, and consider $\gamma$, $S$ and $W$ the chain and the sets of processes that render $\GCLpredicate(m,w,\view)$ true according to Algorithm~\ref{algo:GCL:predicate}.
Consider $w'\in\mathbb{N}^{+}$ such that $w'\leq w$.
Let us note
\begin{align*}
    \gamma_{w}&=\subchain(\gamma,3,t+3-w),\\
    \gamma_{w'}&=\subchain(\gamma,3,t+3-w').
\end{align*}
By definition of \subchain, and because $w'\leq w$, $\gamma_{w}$ is a prefix of $\gamma_{w'}$, which implies
\begin{gather}
    \set(\gamma_{w'})\supseteq \set(\gamma_{w}),\label{eq:gamma:w':gamme:w:inclusion}\\
    \Big|\set(\gamma_{w'})\setminus \set(\gamma_{w})\Big|\leq
    (w-w').\label{eq:diff:subchains:w':w}
\end{gather}
\ta{this part must be better explained}
Define $W'\isdefinedas S\setminus \set(\gamma_{w'})$.
The following holds
\begin{align}
    W'&= \Big(S\setminus \set(\gamma_{w})\Big)\setminus \Big(\set(\gamma_{w'})\setminus \set(\gamma_{w})\Big),\tag{using (\ref{eq:gamma:w':gamme:w:inclusion})}\\
    |W'|&\geq \big|S\setminus \set(\gamma_{w})\big| - (w-w'),\tag{using (\ref{eq:diff:subchains:w':w})}\\
    |W'|&\geq \big|W\big| - (w-w'),\tag{by definition of $W$}\\
    |W'|&\geq w - (w-w') = w'.\tag{since $\GCLpredicate(m,w,\view)$ holds}
\end{align}
This last equation shows that $\GCLpredicate(m,w',\view)$ holds using $\gamma$, $S$, and $W'$ in Algorithm~\ref{algo:GCL:predicate}, concluding the proof.
\end{detailedproof}

\begin{lemma}\label{lemma:GCK:WBPViewMonoProp}
  The pair $(\GCLpredicate,\revealGCL)$ fulfills the \WBPViewMonoProp, \WBPRevealWeightMonoProp, and \WBPLocalConspProp properties defined in Section~{\em\ref{sec:lambd-brb-robustn}}.
\end{lemma}

\begin{proof}~
  \begin{itemize}
  \item \WBPViewMonoProp: Assume $\GCLpredicate(m,w,\view)$ holds, and consider $\gamma$, $S$ and $W$ the chain and the sets of processes that render $\GCLpredicate(m,w,\view)$ true according to Algorithm~\ref{algo:GCL:predicate}.
  If $\view'\supseteq\view$, then we still have $m\cryptchain\gamma\in\view'$.
  If we note
    \begin{align*}
      S'&\isdefinedas\big\{q\in\Pi\mid q\text{ backs $m$ in round $1$ or $2$ in $\view'$}\big\},\\
      W'&\isdefinedas S'\setminus \set\big(\subchain(\gamma,3,t+3-w)\big),
    \end{align*}
    $\view'\supseteq\view$ implies that $S'\supseteq S$ and $W'\supseteq W$, and therefore $|W'|\geq |W|\geq w$, proving the \WBPViewMonoProp property of \GCLpredicate.

\item \WBPRevealWeightMonoProp is trivially fulfilled by the definition of \revealGCL.% $$\revealGCL \isdefinedas \maxfn(2,t+3-w).$$

\item \WBPLocalConspProp: Assume $\GCLpredicate(m,1,\view)$ holds, and consider $\gamma$, $S$ and $W$ the chain and the sets of processes that render $\GCLpredicate(m,1,\view)$ true according to Algorithm~\ref{algo:GCL:predicate}.
Trivially, $m\cryptchain\gamma\in\view$ implies $m\in\messageOf(view)$.

Conversely, assume $m\in\messageOf(view)$.
There exists a chain $\gamma_m$ such that $m\cryptchain\gamma_m\in\view$.
As all chains of $\view$ are valid by assumption, $\gamma_m$ must start with $\psender$'s signature.
Define the sets $S$ and $W$ as in Algorithm~\ref{algo:GCL:predicate}, but using $\gamma_m$.
By construction, $\psender\in S$.
As valid chains are acyclic, $\psender\not\in\set\big(\subchain(\gamma_m,3,t+3-w)\big)$, and therefore $\psender\in W$, and $|W|\geq 1$.
This last statement implies that $\GCLpredicate(m,1,\view)$ holds.
With the previous paragraph, this proves that \GCLpredicate fulfills the \WBPLocalConspProp, and concludes the lemma.
\qedhere
\end{itemize}
\end{proof}

\begin{lemma}\label{lemma:GCL:WBPConspiProp}
  The pair $(\GCLpredicate,\revealGCL)$ fulfills the \WBPConspiProp property.
\end{lemma}

%% FT21Dec22: For reference, the actual property
% \WBPConspiProp: Consider $p_i$, $p_j$ two correct processes different from $\psender$, $w$ a weight, if $p_i$ observes $\delpredicate(m,w,\view_i)$ during its execution, and $p_j$ executes round  $\reveal(w)$, then $m$ is known to $p_j$ at the latest by round $\reveal(w)$.

\begin{proofoverview}
  Let us note $\sequencei$ the chain that renders true $\GCLpredicate(m,w,\view_i)$ for $p_i$ (Alg.~\ref{algo:GCL:predicate}). The proof depends on whether $|\sequencei|=\rmwi$ (the round in which $p_i$ observes the revealing chain $m\cryptchain\sequencei$, cf. Alg.~\ref{algo:GCL:predicate}) occurs before or after the round $\revealGCL(w)\isdefinedas \maxfn(2,t+3-w)$, the round during which we seek to prove that all correct processes that have not stopped earlier are aware of $m$.
  If $\rmwi<\maxfn(2,t+3-w)$, because $p_i$ forwards all chains it has not signed yet, all correct processes observe a chain containing $m$ at the latest by round $\rmwi+1\leq \maxfn(2,t+3-w)$.
  If $\rmwi\geq \maxfn(2,t+3-w)$, the construction of $\revealGCL$ implies that at least $t+1$ processes have signed chains containing $m$ during the first $\maxfn(2,t+3-w)$ rounds of the protocol. One of them must be correct, yielding the lemma.
\end{proofoverview}

\begin{detailedproof}
  Assume a correct process $p_i$ observes 
  $\GCLpredicate(m,w,\view_i)=\ttrue$ during its execution.
  In the following, $\sequencei$ denotes a chain that renders $\GCLpredicate$ true in Algorithm~\ref{algo:GCL:predicate}, and $\rmwi=|\sequencei|$ its length.
  As $p_i$ is correct, $m\cryptchain\sequencei\in\view_i$ (line~\ref{line:m:gamma:in:view} of Algorithm~\ref{algo:GCL:predicate}) implies that $p_i$ receives $m\cryptchain\sequencei$ during round $\rmwi=|\sequencei|$.
  The remainder of the proof distinguishes two cases, depending on whether $\rmwi< \revealGCL(w)\isdefinedas \maxfn(2,t+3-w)$ or not.
  \begin{itemize}
  \item Case $\rmwi< \maxfn(2,t+3-w)$: As just mentioned, $p_i$ receives $m\cryptchain \sequencei$ during the communication step of round $\rmwi$ (line~\ref{line:generic:view:i:R} of Algorithm~\ref{algo:generic:SBRB:pi}).
  If $p_i\not\in\sequencei$, $p_i$ signs the chain (line~\ref{line:generic:signing:next:round}, Alg.~\ref{algo:generic:SBRB:pi}), and broadcasts it during the communication step of round $\rmwi+1 \leq \maxfn(2,t+3-w)$ (line~\ref{line:generic:bcast} of the same algorithm).
  If $p_i\in\sequencei$, $p_i$ has signed a chain $m\cryptchain \seq'$ earlier, and broadcast the result before or during round $\rmwi$.

    In both cases, any correct process $p_j$ that executes round  $\maxfn(2,t+3-w)$ receives some chain $m\cryptchain \sequencei' \cryptchain p_i$ either during or before round $\maxfn(2,t+3-w)$, and therefore $m$ is known to $p_j$ at the latest by round $\maxfn(2,t+3-w)$.
    
  \item Case $\rmwi\geq \maxfn(2,t+3-w)$:
    Let us note $S$ and $W$ the set of processes which together with $\gamma_i$ render $\GCLpredicate(m,w,\view_i)$ true in Algorithm~\ref{algo:GCL:predicate}.

    \newcommand{\subgamma}{\gamma_{3..t+3-w}}
    Let us note $\subgamma\isdefinedas\subchain(\gamma_i,3,t+3-w)$. Since $|\gamma_i|=\rmwi\geq \maxfn(2,t+3-w) \geq t+3-w$, $\gamma_i$ contains more than $t+3-w$ processes. As a result, by definition of the function $\subchain$, and because $\gamma_i$ is valid, and hence acyclic, $\subgamma$ contains exactly $\maxfn(0,(t+3-w)-3+1)=\maxfn(0,t+1-w)$ distinct processes:
    \begin{equation}
      \big|\set(\subgamma\big)|=\maxfn(0,t+1-w). \label{eq:subgamma:t+1-w}
    \end{equation}

    By construction, $W$ and $\set(\subgamma)$ have no element in common, which yields
    \begin{align}
      |W\cup \set(\subgamma)|&=|W|+|\set(\subgamma)|,\\
      &=|W|+\maxfn(0,t+1-w),\tag{using (\ref{eq:subgamma:t+1-w})}\\
      &\geq \maxfn(w,w+t+1-w),\tag{by definition of \GCLpredicate}\\
      |W\cup \set(\subgamma)|&\geq \maxfn(w,t+1) \geq t+1.
    \end{align}

    The set $W\cup \set(\subgamma)$ therefore contains at least one correct process $p_k$. If $p_k\in W$, $p_k$'s signature appears in the first or second position of one of the chains of $\view_i$. If $p_k\in\set(\subgamma)$, $p_k$'s signature appears before position $t+3-w$ in $\gamma_i\in\view_i$. Both cases imply that $p_k$ has signed a chain with message $m$ and has broadcast this chain to all processes that have not stopped earlier during or before round $\maxfn(2,t+3-w)$, proving the lemma.
    \qedhere
    \end{itemize}
  \end{detailedproof}

In the following, we use the notation %% $\view_i[r]$ to denote the chains of $\view_i$ that contain exactly $r$ signatures, and
$\view_i[r_1..r_2]$ the chains of $\view_1$ that contain signatures between $r_1$ and $r_2$.
In the same way $\view_i[r]$ represents the chains received by $p_i$ during round $r$, $\view_i[r_1..r_2]$ contains the chains received between rounds $r_1$ and $r_2$:
\begin{align*}
  \view_i[r_1..r_2] &\isdefinedas \{\gamma\in\view_i: r_1\leq|\gamma|\leq r_2\}.
\end{align*}

Building upon the notation $\view_i[r_1..r_2]$, we introduce the quantity $\accumulatorTwo{i}[R]$ to prove the \WBPVisiProp of $(\GCLpredicate,\revealGCL)$ (Lemma~\ref{lemma:GCL:WBPVisiProp}). $\accumulatorTwo{i}[R]$ is defined using process $p_i$'s view $\view_i$ at round $R$ as follows
%\ft{FT22Dec22: I've re-injected $\truncate_2$ here (but kept $\subchain$ otherwise) as formally $\subchain$ removes the message being signed (since it works on process sequences), but we absolutely want to keep it in \accumulatorTwo{i}. In this part of the paper $\truncate_2$ is also a lighter notation (although it introduces yet another function).}
\begin{equation}
  \accumulatorTwo{i}[R] \isdefinedas \truncate_2\big(\view_i[2..R]\big). \label{eq:def:accumulator}
\end{equation}
%% ${\bigcup_{r'\in[2..R]} \view_i[r']}$ only contains signature chains that are at least of length two. As a result,
$\accumulatorTwo{i}[R]$ contains all length-2 prefixes $m\cryptchain p_{\orig} \cryptchain q$ observed by $p_i$ by round $R$, i.e. % In other words, for a given message $m$, $\accumulatorTwo{i}[R]$ captures 
$p_i$'s knowledge during round $R$ of the  processes that have signed $m$ by the end of round 2.

The following lemma states that all length-2 prefixes known at round $R$ by a correct process $p_i\neq p_{\orig}$ executing Algorithm~\ref{algo:generic:SBRB:pi} are known by all other correct processes by round $R+1$.

\begin{restatable}{lemma}{lemmaAccumIRInJRpOne}
\label{lemma:accum:i:R:in:j:RpOne}
  Let  $p_i$ and $p_j\neq p_i$ be two correct processes,
  such that $p_i$ executes the computation step (lines~{\em\ref{line:generic:start:comm:step}-\ref{line:generic:end:comm:step}}) of at least the $R\leq t$ first rounds, and $p_j$ executes the communication step of at least the first $R+1$ rounds. Then we have 
% \begin{equation*}
    $\forall R\in[1..t]$, $\accumulatorTwo{i}[R]\subseteq \accumulatorTwo{j}[R+1]$.
%  \end{equation*} 
\end{restatable}

\begin{proof}
Note that since $p_i$ and $p_j$ execute Algorithm~\ref{algo:generic:SBRB:pi}, they are both different from $\psender$. We prove the lemma by induction.
\begin{itemize}
  \item Case $R=1$: $\view_i[2..1]=\emptyset$, and therefore $\accumulatorTwo{i}[1]=\emptyset$, trivially proving the case.
  \item Induction case: Let us assume $\accumulatorTwo{i}[R]\subseteq \accumulatorTwo{j}[R+1]$ for some $R\in[1..t-1]$.
  \begin{align*}
      \accumulatorTwo{i}[R+1]&= \truncate_2\big(\view_i[2..R+1]\big),\\
      &=\truncate_2\big(\view_i[2..R]\cup\view_i[R+1]\big),\\
      &=\truncate_2\big(\view_i[2..R]\big)\cup\truncate_2\big(\view_i[R+1]\big),\\
      &=\accumulatorTwo{i}[R]\cup \truncate_2\big(\view_i[R+1]\big),\\
      &\subseteq\accumulatorTwo{j}[R+1]\cup \truncate_2\big(\view_i[R+1]\big), \tag{by case assumption}\\
      &\subseteq\accumulatorTwo{j}[R+2]\cup \truncate_2\big(\view_i[R+1]\big). \tag{as $\accumulatorTwo{j}[R+1]\subseteq\accumulatorTwo{j}[R+2]$}
    \end{align*}
    We now need to show that $\truncate_2\big(\view_i[R+1]\big)\subseteq \accumulatorTwo{j}[R+2]$ to complete the proof. Consider $m\cryptchain p_{\orig} \cryptchain \seq \in \view_i[R+1]$, with $\seq\in \Pi^{R}$. By assumption, $p_i\neq p_{\orig}$, 
    we must therefore distinguish two cases depending whether $p_i$ appears in $\seq$ or not.
    \begin{itemize}
    %% \item Case 1: If $p_i=p_k$, then $p_i$ has broadcast $m\cryptchain p_{\orig} \cryptchain p_i$ to all processes (since $p_i$ is correct) at line~\ref{line:bcast} during the communication step of Round 2, and therefore $m\cryptchain p_{\orig} \cryptchain p_i\in \view_j[2]$, and $\subchain_2(m\cryptchain \seq)\in\subchain_2(\view_j[2])\subseteq \accumulatorTwo{j}[R+2]$.
    \item Case 1: If $p_i\in \seq$, $p_i$ has signed a chain $m\cryptchain p_{\orig} \cryptchain \seq'$ at line~\ref{line:generic:signing:next:round} of Alg.~\ref{algo:generic:SBRB:pi} during a round $R'<R+1$ (where $\seq'\cryptchain p_i$ is a prefix of $\seq$), and $p_i$ has broadcast the chain $m\cryptchain p_{\orig} \cryptchain \seq' \cryptchain p_i$ to all processes (since $p_i$ is correct) at line~\ref{line:generic:bcast} during the communication step of the following round $R'+1\leq R+1$.
    Therefore $m\cryptchain p_{\orig} \cryptchain \seq' \cryptchain p_i\in \view_j[R'+1]$, which implies
    \begin{align*}
      \truncate_2(m\cryptchain p_{\orig} \cryptchain \seq)&=\truncate_2(m\cryptchain p_{\orig} \cryptchain \seq' \cryptchain p_i)\\
      &\in\truncate_2(\view_j[R'+1])\subseteq\accumulatorTwo{j}[R+2].
      \end{align*}
      \item Case 2: If $p_i\not\in \seq$, $p_i$ signs $m\cryptchain p_{\orig} \cryptchain \seq$ during round $R+1$ and as above broadcasts $m\cryptchain p_{\orig} \cryptchain \seq \cryptchain p_i$ at round $R+2$ to all processes. (By construction, the fact that $p_i$ executes the computation step of round $R+1\leq t$ implies that it executes the communication step of round $R+2$.) This similarly implies
        \begin{align*}
          \truncate_2(m\cryptchain p_{\orig} \cryptchain \seq)
          &=\truncate_2(m\cryptchain p_{\orig} \cryptchain \seq \cryptchain p_i)\\
          &\in\truncate_2(\view_j[R+2]) \subseteq\accumulatorTwo{j}[R+2].
        \end{align*}
    \end{itemize}
  \end{itemize}
  These two cases show that $\truncate_2\big(\view_i[R+1]\big)\subseteq \accumulatorTwo{j}[R+2]$, which concludes the proof of the lemma.
\end{proof}

The following lemma shows that if $p_{\orig}$ is Byzantine then all correct processes agree on the length-2 prefixes they have observed by round $t+1$.

\begin{restatable}{lemma}{lemmaTTwoequalFinalRound}\label{lemma:T2:equal:final:round}
  Let $p_{\orig}$ be Byzantine, and $p_i$ and $p_j$ be two correct processes that execute the communication step of round $\finalround$, then $\accumulatorTwo{i}[\finalround]=\accumulatorTwo{j}[\finalround]$.
\end{restatable}

\begin{proofoverview}
  The proof uses the fact that the length-2 prefixes that $p_i$ receives in round $t+1$ have been propagated by $t+1$ processes. One of these processes must be correct, and because $p_{\orig}$ is Byzantine, it must be a process that signed the chain at the earliest in round $2$, implying that the length-2 prefix is also known to $p_j$.
  This observation, together with Lemma~\ref{lemma:accum:i:R:in:j:RpOne} yields the proof. 
\end{proofoverview}

%\lemmaTTwoequalFinalRound*

\begin{detailedproof}
  By definition
  \begin{align}
    \accumulatorTwo{i}[t+1]
    &=\truncate_2\left(\view_i[2..t+1]\right)\nonumber\\
    &=\truncate_2\left(\view_i[2..t]\cup\view_i[t+1]\right),\nonumber\\
    &=\accumulatorTwo{i}[t]\cup \truncate_2\big(\view_i[t+1]\big).\nonumber
  \end{align}
  Applying Lemma~\ref{lemma:accum:i:R:in:j:RpOne} we have $\accumulatorTwo{i}[t]\subseteq \accumulatorTwo{j}[t+1]$, which with the previous equality yields
  \begin{align}
    \accumulatorTwo{i}[t+1] &\subseteq \accumulatorTwo{j}[t+1]\cup \truncate_2\big(\view_i[t+1]\big).\label{eq:T2i:decomp}
  \end{align}

  We now prove that $\truncate_2\big(\view_i[t+1]\big)\subseteq  \accumulatorTwo{j}[t+1]$. Consider $m\cryptchain p_{\orig} \cryptchain \seq \in \view_i[t+1]$. As $p_i$ is correct, it only accepts acyclic signature chains, and $p_{\orig}\not\in\seq$.
  This implies $|\{p_{\orig}\} \cup \set(\seq)|=|\{p_{\orig}\}|+|\set(\seq)|=1+t$.
  So $\{p_{\orig}\} \cup \set(\seq)$ therefore contains at least one correct process, $p_k$. As $p_{\orig}$ is Byzantine by lemma assumption, $p_k\in \seq$, and $p_k$ therefore has signed a chain $m\cryptchain p_{\orig} \cryptchain \seq'$ at line~\ref{line:generic:signing:next:round} of Alg.~\ref{algo:generic:SBRB:pi} before or during round $t$, where $\seq'\cryptchain p_k$ is a prefix of $\seq$. As a result, $p_k$ has broadcast the resulting chain $m\cryptchain p_{\orig} \cryptchain \seq' \cryptchain p_k$ to all other processes during the following round $R'\leq t+1$. This implies $m\cryptchain p_{\orig} \cryptchain \seq' \cryptchain p_k\in \view_j[R']$, and hence
  \begin{align*}
    \truncate_2(m\cryptchain p_{\orig} \cryptchain \seq)&=\truncate_2(m\cryptchain p_{\orig} \cryptchain \seq' \cryptchain p_k)\\
    &\in\truncate_2(\view_j[R']) \subseteq\accumulatorTwo{j}[t+1].
  \end{align*}
  This last equation shows that $\truncate_2(\view_i[t+1])\subseteq\accumulatorTwo{j}[t+1]$, which injected in (\ref{eq:T2i:decomp}) yields $\accumulatorTwo{i}[t+1]\subseteq \accumulatorTwo{j}[t+1]$. 
  By inverting $p_i$ and $p_j$, by the same reasoning we obtain $\accumulatorTwo{j}[t+1]\subseteq \accumulatorTwo{i}[t+1]$, which concludes the Lemma's proof.
\end{detailedproof}

\begin{corollary}\label{coro:T2:i:r:in:T2:j:final}
Let $p_{\orig}$ be Byzantine, $p_i$ and $p_j$ be two correct processes, such that $p_i$ executes the computation step of at least the first $r\in[1..\finalround]$ rounds, and $p_j$ executes the communication step of all $\finalround$ rounds. Then we have $\accumulatorTwo{i}[r]\subseteq \accumulatorTwo{j}[\finalround]$.
\end{corollary}

\begin{proof}The proof follows either from Lemma~\ref{lemma:T2:equal:final:round} or~\ref{lemma:accum:i:R:in:j:RpOne}, depending on whether $r=\finalround$ or not.
  \begin{itemize}
  \item If $r=\finalround$, the corollary follows trivially from Lemma~\ref{lemma:T2:equal:final:round}.
  \item If $r<\finalround$, this follows from Lemma~\ref{lemma:accum:i:R:in:j:RpOne}, and observing that $\accumulatorTwo{j}[r+1]\subseteq \accumulatorTwo{j}[\finalround]$. \qedhere % by definition of $\accumulatorTwo{i}[\cdot]$. 
  \end{itemize}
\end{proof}

%% FT21Dec22: For the record, definition of WBP final visibility
% \WBPVisiProp: Consider $p_i$, $p_j$ two correct processes different from $\psender$, if $\psender$ is Byzantine and $p_i$ observes $\delpredicate(m,$ $w,$ $\view_i)$ during its execution, and $p_j$ executes round $t+1$, then $p_j$ observes $\delpredicate(m,$ $w,$ $\view_j)$ at round $t+1$.

\begin{lemma}\label{lemma:GCL:WBPVisiProp}
  The pair $(\GCLpredicate,\revealGCL)$ fulfills the \WBPVisiProp property.
\end{lemma}

\begin{proofoverview}
  Consider $p_i$, $p_j$ two correct processes, and assume \psender is Byzantine. The proof focuses on the sets of processes $S$ and $W$, and on the revealing chain $\gamma$ used in the predicate \GCLpredicate (Algorithm~\ref{algo:GCL:predicate}).
  First, the proof shows that the set $S$ perceived by $p_i$ is propagated to all other correct processes at the latest by round $t+1$. More concretely, if $p_i$ perceives a process $p_k$ as backing $m$, then all correct processes also perceive $p_k$ as backing $m$ by round $t+1$.
  The proof then shows that any revealing chain $\gamma_i$ that renders \GCLpredicate $\ttrue$ for $p_i$ implies the existence of a revealing chain $\gamma_j$ is $p_j$'s view at round $t+1$ that is no more ``constraining'' than $\gamma_i$.
\end{proofoverview}

\begin{detailedproof}
  Assume \psender is Byzantine.
  Assume a correct process $p_i$ observes $\GCLpredicate(m,$ $w,$ $\view_i)=$ $\ttrue$ during some round $R$ of its execution. Consider $p_j$ another correct process that reaches round $\finalround$. Without loss of generality, assume $p_j\neq p_i$ (as the case $p_i=p_j$ is trivial).

  In the following, for clarity, we note $\view_x^y$ the value of the $\view$ variable in Algorithm~\ref{algo:generic:SBRB:pi} for the process $p_x$ at round $y$ (or more precisely, after \view has been updated in round $y$ at line~\ref{line:generic:view:i:R} of Algorithm~\ref{algo:generic:SBRB:pi}).
  
  Consider $\sequencei$, $S_{i,R}$ and $W_{i,R}$ the chain and the sets of processes that render true $\GCLpredicate$ for $p_i$ in Algorithm~\ref{algo:GCL:predicate} during round $R$. %$\rmwi=|\sequencei|$ is the round at which $p_i$ has received $m\cryptchain \sequencei$.

By definition of $\GCLpredicate$, $|W_{i,R}|\geq w$.
Using $\accumulatorTwo{i}[R]$ (\ref{eq:def:accumulator}), we can express $S_{i,R}$ as follows
\ft{There is a hidden subtlety here: $\accumulatorTwo{i}[R]$ is defined based on the value of $\view_i$ at round $R$ (or more precisely of $\view_i$ just after the communication step of round $R$.)}
\begin{align*}
  S_{i,R}&=\{\set(\gamma) \mid m\cryptchain\gamma\in \view_i^R[1]\cup\accumulatorTwo{i}[R]\}.
\end{align*}
The set $\view_i^R[1]$ contains the chains received by $p_i$ in round $1$, while $\accumulatorTwo{i}[R]$ contains the length-2 prefixes of chains received by $p_i$ in rounds $2$ to $R$.
If $m\cryptchain\gamma\in \view_i^R[1]\cup\accumulatorTwo{i}[R]$, the processes whose signature appears in $\gamma$ have therefore backed $m$ in round $1$ or $2$ in $\view_i^R$.
\newcommand{\subseqi}{\gamma^i_{3..t+3-w}}
Let us define $\subseqi\isdefinedas\subchain(\sequencei,3,t+3-w)$.
We can express $W_{i,R}$ as
\begin{align*}    
  W_{i,R}&=S_{i,R}\setminus \set\big(\subseqi\big)\\
  &=\big\{\set(\gamma) \mid m\cryptchain\gamma\in \view_i^R[1]\cup\accumulatorTwo{i}[R]\big\}\, \setminus \,\set\big(\subseqi\big).
\end{align*}

In the following, we will consider $p_j$'s perception at round $t+1$ of the processes that have backed $m$ in round $1$ or $2$. The set of these processes can be defined as:
\begin{equation*}
  S_{j,t+1}=\{\set(\gamma) \mid m\cryptchain\gamma\in \view_j^{t+1}[1]\cup\accumulatorTwo{j}[t+1]\}.
\end{equation*}

\begin{sublemma}\label{sublemma:S:i:R:subset:S:j:t+1}
  $S_{i,R}\subseteq S_{j,t+1}.$
\end{sublemma}

{%
\renewcommand{\proofname}{Proof}
\begin{proof}
  Consider $q\in S_{i,R}$. There exists a chain $\gamma_q$ such that $q\in\gamma_q$ and $m\cryptchain\gamma_q\in \view_i^R[1]\cup\accumulatorTwo{i}[R]$.
  \begin{itemize}
  \item If $m\cryptchain\gamma_q\in \accumulatorTwo{i}[R]$, by Corollary~\ref{coro:T2:i:r:in:T2:j:final}, $\accumulatorTwo{i}[R]\subseteq \accumulatorTwo{j}[\finalround]$, and $\gamma_q\in \accumulatorTwo{j}[\finalround]$, which leads to $q\in S_{j,t+1}$.

  \item If $m\cryptchain\gamma_q\in\view_i^R[1]$, then because $p_i$ is correct, $m\cryptchain\gamma_q$ is a valid chain of length 1. As a result, $\gamma_q=(\psender)$ and $q=\psender$. Since $p_i\neq \psender$, $p_i$ signs $m\cryptchain\psender$ at line~\ref{line:generic:signing:next:round} of Alg.~\ref{algo:generic:SBRB:pi}, and broadcasts  $m\cryptchain\psender\cryptchain p_i$ to all correct processes in round $2$. We therefore have $m\cryptchain\psender\cryptchain p_i\in \accumulatorTwo{j}[2]\subseteq \accumulatorTwo{j}[\finalround]$ (assuming $t\geq 1$, and $t+1\geq 2$). As a result, $q=\psender\in S_{j,t+1}$, which concludes the lemma.
    \qedhere
  \end{itemize}
\end{proof}
}%

To conclude the proof, we need to find a chain $m\cryptchain\gamma_j\in\view_j^{t+1}$ that, when ``removed'' from $S_{j,t+1}$, produces a set $W_{j,t+1}$ containing at least $w$ processes. Once this happens, $m\cryptchain\gamma_j$, $S_{j,t+1}$ and $W_{j,t+1}$ will render $\GCLpredicate(m,$ $w,$ $\view_j^{t+1})$ true in round $t+1$.
The existence of $m\cryptchain\gamma_j$ comes from the following sublemma.
% To find $m\cryptchain\gamma_j$, we distinguish three cases: the case where  $p_i\in S_{i,R}$, the case where $R=t+1$, and the remaining cases.

\begin{sublemma}\label{sublemma:exists:gamma:i}
  $\Exists m\cryptchain\gamma_j\in\view_j^{t+1}, \Big[\set\big(\subchain(\gamma_j,3,t+3-w)\big)\cap S_{i,R}\Big] \,\subseteq\, \set(\subseqi).$
\end{sublemma}

{%
\renewcommand{\proofname}{Proof}
\begin{proof}
  The proof distinguishes three cases:
  \begin{itemize}
  \item Case 1: $p_i\in S_{i,R}$. Because $p_i\neq\psender$ and $p_i$ is correct, if $p_i\in S_{i,R}$, then $p_i$ has signed a chain $m\cryptchain \psender \cryptchain p_i$, and broadcast this chain in round $2$ to all other correct processes, including $p_j$. As a result, $m\cryptchain \psender \cryptchain p_i\in\view_j^{t+1}$ (assuming $t\geq 1$).

    Furthermore, since $m\cryptchain \psender \cryptchain p_i$ only contains 2 processes, $$\set\big(\subchain(\psender \cryptchain p_i,3,t+3-w)\big)=\emptyset,$$ which proves the sublemma.
  \item Case 2: $|\gamma_i|=t+1$.
  In this case, because $\gamma_i$ is acyclic, it contains at least one correct process $p_k$.
  Since $p_k$ is correct, it has signed and then broadcast a chain $m\cryptchain \gamma_k \cryptchain p_k$ to all processes including $p_j$ at the latest by round $t+1$. As a result, $m\cryptchain \gamma_k \cryptchain p_k\in \view_j^{t+1}$.

    Furthermore, by construction, $\gamma_k \cryptchain p_k$ is a prefix of $\gamma_i$, which implies
    \begin{equation*}
      \set\big(\subchain(\gamma_k \cryptchain p_k,3,t+3-w)\big)\subseteq \set(\subseqi),
    \end{equation*}
    which proves the sublemma.
    
  \item Case 3: $p_i\not\in S_{i,R} \wedge |\gamma_i|<t+1$.
    If $p_i\in \gamma_i$, then we can reuse the reasoning of Case 2 with $p_k=p_i$, which proves the sublemma. If $p_i\not\in \gamma_i$, since $|\gamma_i|<t+1$, $p_i$ has received $\gamma_i$ in a round $r_i=|\gamma_i|$, and because it is correct, $p_i$ has signed and broadcast the chain $m\cryptchain\gamma_i\cryptchain p_i$ in round $r_i+1\leq t+1$ to all processes, including $p_j$. As a result $m\cryptchain\gamma_i\cryptchain p_i\in\view_j^{t+1}$.

    Furthermore
    \begin{align*}
      \set\big(\subchain(\gamma_i\cryptchain p_i,3,t+3-w)\big)
      &\subseteq \set\big(\subchain(\gamma_i,3,t+3-w)\big)\cup\{p_i\},\\
      &\subseteq \set(\subseqi)\cup\{p_i\}.\tag{by definition of $\subseqi$}
    \end{align*}
    As a result
     \begin{align*}
       \set\big(\subchain(\gamma_i\cryptchain p_i,3,t+3-w)\big)\cap S_{i,R},
       &\subseteq \big(\set(\subseqi)\cap S_{i,R}\big) \cup \big(\{p_i\}\cap S_{i,R}\big),\\
       &\subseteq \:\set(\subseqi)\cap S_{i,R}.\tag{since $p_i\not\in S_{i,R}$ by case assumption}
     \end{align*}
     This last inclusion concludes Case 3 and therefore the sublemma.\qedhere
  \end{itemize}
\end{proof}
}%
\noindent
In the following, we will note $m\cryptchain\gamma_j$ the chain whose existence is given by Sublemma~\ref{sublemma:exists:gamma:i}. Consider
\begin{equation*}
  W_{j,t+1}= S_{j,t+1} \setminus \set\big(\subchain(\gamma_j,3,t+3-w)\big).
\end{equation*}
The following holds
\begin{align*}
  W_{i,R} &= S_{i,R}\setminus \set\big(\subseqi\big)\\
  &\subseteq S_{i,R}\setminus \Big(\set\big(\subchain(\gamma_j,3,t+3-w)\big)\cap S_{i,R}\Big)\tag{using Sublemma~\ref{sublemma:exists:gamma:i}}\\
  &\subseteq S_{i,R}\setminus \set\big(\subchain(\gamma_j,3,t+3-w)\big)\\
  &\subseteq S_{j,t+1}\setminus \set\big(\subchain(\gamma_j,3,t+3-w)\big)\tag{using Sublemma~\ref{sublemma:S:i:R:subset:S:j:t+1}}\\
  &\subseteq W_{j,t+1}.
\end{align*}

Because $|W_{i,R}|\geq w$ by definition, this last inclusion yields $|W_{j,t+1}|\geq w$. This inequality shows that $\GCLpredicate(m,w,\view_j)$ is rendered true at $p_j$ at round $t+1$ by using $\gamma_j$, $S_{j,t+1}$, and $W_{j,t+1}$ in Algorithm~\ref{algo:GCL:predicate}, thus proving the lemma.
\qedhere
\end{detailedproof}

%% FT23Dec22
% If $\psender$ is correct and \brb-broadcasts a message $m$, all correct processes $p_i\neq\psender$ observe $\delpredicate(m,\wgood,\view_i)$ during round 2, where $\wgood\in\mathbb{N}^{+}$ is a weight such that $\reveal(\wgood)=\lambdagood$.

\begin{lemma}\label{lemma:GCL:WBPLiveProp}
  The pair $(\GCLpredicate,\revealGCL)$ fulfills the \WBPLiveProp property.
\end{lemma}

\begin{proof}
  When the initial sender, $\psender$, is correct, and \brb-broadcast a message $m$, all remaining correct processes $p_i\neq\psender$ receive $m\cryptchain \psender$ in round $1$, and broadcast $m\cryptchain \psender \cryptchain p_i$ in round $2$. As a result, every correct process $p_i\neq\psender$ receives at least $c-1$ distinct chains backing $m$ of the form $m\cryptchain \psender \cryptchain p_k$ in round $2$.

  These chains imply that the set $S$ at line~\ref{line:certif:set:S} of Algorithm~\ref{algo:GCL:predicate} contains at least $c$ processes ($\psender$ and the remaining $c-1$ correct processes). Using any of the length-2 chains $m\cryptchain \psender \cryptchain p_k$, a set $W$ can be constructed equal to $S$, rendering $\GCLpredicate(m,c,\view_i)$ true for all correct processes other than $\psender$.

  The observation that $\revealGCL(c)=\maxfn(2,t+3-c)$ concludes the lemma.
\end{proof}

\subsection{Proof of Theorem~\ref{theo:GCL:t+3-c:BRB:robust} and Corollary~\ref{coro:GCL:BRB:goodcase}}

The proof of Theorem~\ref{theo:GCL:t+3-c:BRB:robust} follows from Lemmas~\ref{lemma:GCL:WBPWeightMonoProp}, \ref{lemma:GCK:WBPViewMonoProp}, \ref{lemma:GCL:WBPConspiProp}, \ref{lemma:GCL:WBPVisiProp}, and~\ref{lemma:GCL:WBPLiveProp}.
The proof of Corollary~\ref{coro:GCL:BRB:goodcase} follows from Theorems~\ref{theo:algo:generic:SBRB:pi:works} and~\ref{theo:GCL:t+3-c:BRB:robust}.
\section{Conclusion}
Considering $n$-process synchronous distributed systems where up to $t<n$ processes  can be Byzantine, this paper explored the good-case latency of deterministic Byzantine reliable broadcast (BRB) algorithms, namely the time taken by correct processes to deliver a message when the initial sender is correct.

In contrast to their randomized counterparts, no deterministic BRB algorithm was known that exhibits a good-case latency better than $t+1$ (the worst-case bound) under a majority of Byzantine processes.
This paper has proposed a novel deterministic synchronous BRB algorithm that substantially improves on this earlier bound and provides a good case latency of $\maxfn(2,t+3-c)$ rounds, where $t$ is the upper bound on the number of Byzantine processes, and $c$ the number of effectively correct processes in the considered run.

The proposed algorithm has been presented as an instance of a \textit{generic} BRB algorithm (from which the classical BRB algorithm from Lamport, Shostak and Pease~\cite{LSP82} can also be derived).
This generic algorithm extends the ``signature chain mechanism'' first proposed four decades ago. It exploits a family of weight-based predicates to identify patterns in sets of signature chains and help correct processes decide when they can safely deliver a message. A judicious choice for these patterns delivers a concrete BRB algorithm that allows correct processes to \brb-deliver much earlier than earlier proposals when the context is favorable.
In particular, when the sender is correct, and there are enough effectively correct processes ($c>t$), the resulting algorithm delivers in $2$ rounds, thus outperforming all known dishonest-majority BRB algorithms (whether deterministic or randomized).

Several crucial open questions remain, in particular, whether the upper bound of $\maxfn(2,t+3-c)$ rounds can be further improved (for instance, using techniques 
employed in sub-linear randomized algorithms~\cite{ANRX21-2,WXSD20}).
In terms of lower bounds, one might ask whether the lower bound of $\lfloor n / (n-t) \rfloor -1$ shown in \cite{ANRX21-2} can be refined to include the effective number of correct processes $c$, and whether this same lower bound can be strengthened in the deterministic case, for instance by observing Byzantine Agreement cannot be solved in a (worst-case) sub-linear communication complexity using algorithms that tolerate a strongly adaptive adversary (which include deterministic algorithms)~\cite{ACDNPRS19}.

Finally, this paper did not consider the problem of communication complexity.
Signature chains can be quite costly in a practical implementation, as each new signature  adds hundreds or even thousands of bits to each network message.
We conjecture that multi-signatures could help to significantly reduce this overhead by aggregating all non-backing signatures into a fixed-size structure~\cite{BLS01}.
Furthermore, an interesting follow-up would then be to study the tension between time complexity and communication complexity and how favoring one metric may hinder the other.

%\section*{Acknowledgments}
%This work was partially supported by the French ANR project ByBLoS (ANR-20-CE25-0002-01), and by the PriCLeSS project granted by the Labex CominLabs excellence laboratory of the French ANR (ANR-10-LABX-07-01).
%The authors would like to thank the anonymous reviewers whose careful reading and suggestions helped them improve their paper.
%
\section*{Declarations}
\paragraph{Funding.}
This work was partially funded by the French ANR project ByBLoS (ANR-20-CE25-0002-01), and by the PriCLeSS project granted by the Labex CominLabs excellence laboratory of the French ANR (ANR-10-LABX-07-01).
%The authors would like to thank the anonymous reviewers whose careful reading and suggestions helped them improve their paper.

\paragraph{Competing Interests.}
The authors have no competing interests to declare that are relevant to the content of this article, apart from those possibly resulting from the funding sources mentioned above.
\paragraph{Data.}
 Data sharing not applicable to this article as no datasets were generated or analyzed during the current study.%

\bibliographystyle{plain}
\bibliography{bibliography}

\begin{thebibliography}{10}

\bibitem{ACDNPRS19}
Ittai Abraham, T-H.~Hubert Chan, Danny Dolev, Kartik Nayak, Rafael Pass, Ling
  Ren, and Elaine Shi.
\newblock Communication complexity of {Byzantine} agreement, revisited.
\newblock In {\em ACM Symposium on Principles of Distributed Computing (PODC)},
  pages 317--326, 2019.

\bibitem{AMNRY20}
Ittai Abraham, Dahlia Malkhi, Kartik Nayak, Ling Ren, and Maofan Yin.
\newblock Sync {HotStuff}: Simple and practical synchronous state machine
  replication.
\newblock In {\em IEEE Symposium on Security and Privacy (S\&P)}, pages
  106--118, 2020.

\bibitem{ANRX21}
Ittai Abraham, Kartik Nayak, Ling Ren, and Zhuolun Xiang.
\newblock Good-case latency of {Byzantine} broadcast: A complete
  categorization.
\newblock In {\em ACM Symposium on Principles of Distributed Computing (PODC)},
  pages 331--341, 2021.

\bibitem{ANRX21-2}
Ittai Abraham, Kartik Nayak, Ling Ren, and Zhuolun Xiang.
\newblock Good-case latency of {Byzantine} broadcast: A complete
  categorization.
\newblock In {\em arXiv:2102.07240}, pages 1--38, 2021.

\bibitem{AW04}
Hagit Attiya and Jennifer~L. Welch.
\newblock {\em Distributed computing - fundamentals, simulations, and advanced
  topics (2. ed.)}.
\newblock Wiley series on parallel and distributed computing. Wiley, 2004.

\bibitem{AFRT20}
Alex Auvolat, Davide Frey, Michel Raynal, and Fran{\c{c}}ois Ta{\"{\i}}ani.
\newblock Money transfer made simple: A specification, a generic algorithm and
  its proof.
\newblock {\em Bulletin of {EATCS}}, 132:22--43, 2020.

\bibitem{AFRT21}
Alex Auvolat, Davide Frey, Michel Raynal, and Fran{\c{c}}ois Ta{\"{\i}}ani.
\newblock Byzantine-tolerant causal broadcast.
\newblock {\em Theoretical Computer Science}, 885:55--68, 2021.

\bibitem{BDS20}
Mathieu Baudet, George Danezis, and Alberto Sonnino.
\newblock {F}ast{P}ay: High-performance {Byzantine} fault tolerant settlement.
\newblock In {\em ACM Advances in Financial Technologies}, pages 163--177,
  2020.

\bibitem{BLS01}
Dan Boneh, Ben Lynn, and Hovav Shacham.
\newblock Short signatures from the {Weil} pairing.
\newblock In {\em International Conference on the Theory and Application of
  Cryptology and Information Security}, pages 514--532. Springer, 2001.

\bibitem{boneh2001short}
Dan Boneh, Ben Lynn, and Hovav Shacham.
\newblock Short signatures from the weil pairing.
\newblock In {\em 7th International Conference on the Theory and Application of
  Cryptology and Information Security (ASIACRYPT 2001)}, pages 514--532.
  Springer, 2001.

\bibitem{B87}
Gabriel Bracha.
\newblock Asynchronous {Byzantine} agreement protocols.
\newblock {\em Information \& Computation}, 75(2):130--143, 1987.

\bibitem{CGL11}
Christian Cachin, Rachid Guerraoui, and Lu{\'{\i}}s E.~T. Rodrigues.
\newblock {\em Introduction to Reliable and Secure Distributed Programming (2.
  ed.)}.
\newblock Springer, 2011.

\bibitem{CM19}
Jing Chen and Silvio Micali.
\newblock Algorand: A secure and efficient distributed ledger.
\newblock {\em Theoretical Computer Science}, 777:155--183, 2019.

\bibitem{CGKKMPPST20}
Daniel Collins, Rachid Guerraoui, Jovan Komatovic, Petr Kuznetsov, Matteo
  Monti, Matej Pavlovic, Yvonne{-}Anne Pignolet, Dragos{-}Adrian Seredinschi,
  Andrei Tonkikh, and Athanasios Xygkis.
\newblock Online payments by merely broadcasting messages.
\newblock In {\em Dependable Systems and Networks (DSN)}, pages 26--38. IEEE,
  2020.

\bibitem{DRS90}
Danny Dolev, Ruediger Reischuk, and H~Raymond Strong.
\newblock Early stopping in {Byzantine} agreement.
\newblock {\em Journal of the ACM}, 37(4):720--741, 1990.

\bibitem{DS83}
Danny Dolev and H.~Raymond Strong.
\newblock Authenticated algorithms for {Byzantine} agreement.
\newblock {\em SIAM Journal on Computing}, 12(4):656--666, 1983.

\bibitem{DLS88}
Cynthia Dwork, Nancy~A. Lynch, and Larry~J. Stockmeyer.
\newblock Consensus in the presence of partial synchrony.
\newblock {\em Journal of the {ACM}}, 35(2):288--323, 1988.

\bibitem{FN09}
Matthias Fitzi and Jesper~Buus Nielsen.
\newblock On the number of synchronous rounds sufficient for authenticated
  {Byzantine} agreement.
\newblock In {\em International Symposium on Distributed Computing (DISC)},
  Springer LNCS 5805, pages 449--463, 2009.

\bibitem{DBLP:conf/pact/FreyGRT21}
Davide Frey, Lucie Guillou, Michel Raynal, and Fran{\c{c}}ois Ta{\"{\i}}ani.
\newblock Consensus-free ledgers when operations of distinct processes are
  commutative.
\newblock In {\em 16th International Conference on Parallel Computing
  Technologies (PaCT)}, Springer LNCS 12942, pages 359--370, 2021.

\bibitem{GKQV10}
Rachid Guerraoui, Nikola Knezevic, Vivien Qu{\'{e}}ma, and Marko Vukolic.
\newblock The next 700 {BFT} protocols.
\newblock In {\em EuroSys}, pages 363--376. {ACM}, 2010.

\bibitem{GKMPS22}
Rachid Guerraoui, Petr Kuznetsov, Matteo Monti, Matej Pavlovic, and
  Dragos{-}Adrian Seredinschi.
\newblock The consensus number of a cryptocurrency.
\newblock {\em Distributed Computing}, 35(1):1--15, 2022.

\bibitem{IR16}
Damien Imbs and Michel Raynal.
\newblock Trading off \emph{t}-resilience for efficiency in asynchronous
  {Byzantine} reliable broadcast.
\newblock {\em Parallel Processing Letters}, 26(4):1650017:1--1650017:8, 2016.

\bibitem{LSP82}
Leslie Lamport, Robert~E. Shostak, and Marshall~C. Pease.
\newblock The {Byzantine} generals problem.
\newblock {\em ACM Transactions on Programming Languages and Systems},
  4(3):382--401, 1982.

\bibitem{MRV99}
Silvio Micali, Michael~O. Rabin, and Salil~P. Vadhan.
\newblock Verifiable random functions.
\newblock In {\em 40th IEEE Symposium on the Foundations of Computer Science
  (FOCS)}, pages 120--130, 1999.

\bibitem{MHR14}
Achour Most{\'{e}}faoui, Moumen Hamouma, and Michel Raynal.
\newblock Signature-free asynchronous {Byzantine} consensus with $t<n/3$ and
  {$O(n^2)$} messages.
\newblock In {\em ACM Symposium on Principles of Distributed Computing (PODC)},
  pages 2--9. {ACM}, 2014.

\bibitem{NRSVX20}
Kartik Nayak, Ling Ren, Elaine Shi, Nitin~H. Vaidya, and Zhuolun Xiang.
\newblock Improved extension protocols for {Byzantine} broadcast and agreement.
\newblock In {\em International Symposium on Distributed Computing (DISC)},
  volume 179 of {\em LIPIcs}, pages 28:1--28:17, 2020.

\bibitem{PSL80}
Marshall~C. Pease, Robert~E. Shostak, and Leslie Lamport.
\newblock Reaching agreement in the presence of faults.
\newblock {\em Journal of the {ACM}}, 27(2):228--234, 1980.

\bibitem{R18}
Michel Raynal.
\newblock {\em Fault-Tolerant Message-Passing Distributed Systems - An
  Algorithmic Approach}.
\newblock Springer, 459 pages, 2018.

\bibitem{WXDS20}
Jun Wan, Hanshen Xiao, Srinivas Devadas, and Elaine Shi.
\newblock Round-efficient {Byzantine} broadcast under strongly adaptive and
  majority corruptions.
\newblock In {\em 18th Theory of Cryptography Conference (TCC)}, Springer LNCS
  12550, pages 412--456, 2020.

\bibitem{WXSD20}
Jun Wan, Hanshen Xiao, Elaine Shi, and Srinivas Devadas.
\newblock Expected constant round {Byzantine} broadcast under dishonest
  majority.
\newblock In {\em 18th Theory of Cryptography Conference (TCC)}, Springer LNCS
  12550, pages 381--411, 2020.

\end{thebibliography}

\newpage
\appendix
\section{Communication Complexity}

Because our algorithm focuses on latency, we have left aside communication complexity so far, but this would be the next aspect to consider.
First off, it should be noted that the communication cost of a distributed algorithm can be approached from two perspectives, by considering two distinct metrics:

\begin{itemize}
    \item Metric 1: The number of messages sent by correct processes,
    \item Metric 2: The size of individual messages sent by correct processes.
\end{itemize}

These two metrics only measure the communication cost of correct processes, as Byzantine processes can send a potentially infinite amount of messages with a potentially infinite amount of information.
The two metrics can be combined in one single metric: the total amount of information sent in the network by correct processes (Metric 3).
But to approach the problem in a more granular manner, it is useful to consider these 2 metrics separately.
Moreover, in a practical implementation, it is often more desirable to reduce the number of messages (Metric 1) than the size of messages (Metric 2), as each new message often necessitates negotiating a new network connection (e.g., a TCP session between two endpoints) which entails a significant communication overhead.
In other words, having a few large messages is often better than having a lot of small messages.
In this respect, Metric 3 does not allow a fine-grained analysis of communication complexity.

In the following, ``the algorithm'' refers to Algorithm~\ref{algo:generic:SBRB:pi} in which we use the $\GCLpredicate$ predicate described in Algorithm~\ref{algo:GCL:predicate}.

\subsection{Metric 1: Number of messages}

\newcommand{\Rend}{\ensuremath{R_{\mathit{end}}}\xspace}

Currently, the algorithm follows the original BRB algorithm of Lamport, Shostak, and Pease and adopts a full-knowledge dissemination strategy: correct processes forward any signature chain they have not signed yet.
This is not very efficient, and as a result, in each round, each correct process sends $n$ messages (1 unreliable broadcast).
For simplicity's sake, we will ignore the border case where a process cannot sign any new chain.
This assumption leads to a worst-case collective message cost for correct processes of $n$ messages in round 1, and $c \times n$ messages in each later round, which adds up to $n+ c \times n \times \Rend$ messages (at most), where \Rend is the maximum number of rounds needed for the algorithm to brb-deliver messages at correct processes.\footnote{When $\Rend\leq t$, Algorithm~\ref{algo:generic:SBRB:pi} \brb-delivers in round \Rend, but also executes one extra round of communication (through the $\ready_i$ variable). This extra round leads to a total of at most $\Rend+1$ rounds of communication: the first round followed by $\Rend$ subsequent rounds.\label{fn:Rend:extra:round}}
This upper bound can be further refined depending on whether the execution occurs in a good or bad case.

\subsubsection{Good case (correct sender)}
In a good case, we have $\Rend = \maxfn(2,t+3-c)$, but the initial sender (which is correct) only participates in the first round. These two observations lead to at most $n + (c-1) \times n \times \maxfn(2,t+3-c)$ messages sent by correct processes, which boils down to $n + 2\times n \times (n-1) = O(n^2)$ messages when all processes are correct ($c=n$, assuming $n>t$).
%% FT06Mar23 STOPPED HERE

\subsubsection{Bad case (Byzantine sender)}
In a bad case, we have at worst $\Rend=t+1$. When $\Rend=t+1$, however, correct processes do not execute any extra round of communication (see Footnote~\ref{fn:Rend:extra:round}). This, therefore, leads to at most $n + (c-1) \times n \times (t+1)$ messages sent by correct processes. If one assumes $c$ and $t$ are of the order of $n$, the overall message complexity is thus in $O(n^3)$.

\subsection{Metric 2: Size of messages}
We capture the size of a message by counting how many ``information items'' it contains.
Because the size of the fixed-size fields of a message (e.g., its application payload) becomes asymptotically negligible compared to its fields of variable size (e.g., a set of signatures), we equate this number of ``information items'' by the number of elements in the fields of variable size.
In the case of our algorithm, this number of ``information items'' would be the number of signatures in all of the chains contained in a message.

In our algorithm, message size contributes heavily to communication costs.
This is because messages keep growing in size at every round, as both the number of chains in each message and the length of each chain keep increasing.
Let us, therefore, count the number of signatures exchanged in each round.
A chain exchanged in round $R$ must contain $R$ signatures.
A process $p_i$ can have no more than $C_{n-2}^{R-2}$ (the binomial coefficient ``$n-2$ choose $R-2$'') such chains that need to be disseminated (since these chains must start by \psender's signature, must end with $p_i$'s signature, and are acyclic).

\subsubsection{Good case (correct sender)}
In a good case, in round $R \geq 2$, the chains of length $R$ that a correct process needs to disseminate contain no more than $R \times C_{n-2}^{R-2}$ signatures, and correct processes collectively cannot send more than $(c-1) \times n \times R \times C_{n-2}^{R-2}$ signatures, which is upper bounded by $(c-1) \times n \times R \times (n-2)^{R-2}$.
When all processes are correct ($c=n$, which implies that the algorithm terminates in 3 rounds: 2 for all processes to brb-deliver, plus one for broadcasting a final message, see Footnote~\ref{fn:Rend:extra:round}), this adds up to no more than $n + 2n(n-1) + 3n(n-1)(n-2) = O(n^3)$ signatures.

\subsubsection{Bad case (Byzantine sender)}
In a bad case, given that the Byzantine sender may spam the system with an arbitrary number of application messages (the BRB's payload), with each message incurring its own cost in signatures, the signature cost is essentially unbounded.

\subsection{Possible improvements}

\subsubsection{Sets instead of chains of signatures}
The weight-based predicate $\GCLpredicate$ (Algorithm~\ref{algo:GCL:predicate}, Section~\ref{sec:wbp-gcl-predicate}) does not use the order of the signatures in chains from position 3 to the end.
Indeed, the predicate hinges on the cardinality of the set of signatures $W$ (rather than some sequence of signatures), and the construction of $W$ depends on the interplay of signatures present in the first two rounds and those present from round 3 and later.
%% the only signatures that constitute the ``weight'' of a set of signature chains are the signature at position 1 (the signature of \psender) and the signatures at position 2.
As a result, instead of a chain, every signature from position 3 to the end could be bundled in a set that does not preserve the order.

Using sets of signatures instead of chains would greatly reduce the size of the messages in good and bad cases (Metric 2), as currently, a single message that is exchanged in this algorithm can contain multiple chains that possess exactly the same signatures, but not in the same order, which creates a significant amount of redundant information in each message.

Moreover, using sets of signatures would enable the use of multisignature schemes (such as BLS~\cite{boneh2001short}), which make it possible to aggregate multiple digital signatures in one single fixed-size structure, thus further reducing message sizes.
An important caveat is that, even with multisignatures, the message must still contain the identity of all signatories to verify its authenticity. As a result, although the number of ``authenticators''~\cite{AMNRY20} is reduced, the number of ``information items'' remains the same.

\subsubsection{Filtering received messages}
To limit the size of the message when the sender is Byzantine, we conjecture that a filtering mechanism precluding multiple conflicting application messages from being backed by correct processes could help to obtain a finite value for Metric 2 in a bad case.
This improvement may also lead to a decrease of Metric 1.

\subsubsection{Not sending ``useless'' messages}
There is a strong possibility that some messages of the algorithm that we present are not required to ensure its correctness.
For instance, in our algorithm, a correct process signs and rebroadcasts every signature chain it receives that does not already contain its signature.
However, the set of chains that it broadcasts may be a version with more signatures of an old set of chains that it has already broadcast in a previous round.
As a result, broadcasting this new set of chains to the network may not decrease the revealing round for the message for any process, and not broadcasting this new ``useless'' message may very well have no consequence on the correct termination of the algorithm.
This would reduce Metric 1 (the number of messages) in both good and bad cases.

\end{document}